\documentclass[letterpaper, 10 pt, conference]{ieeeconf}%

\IEEEoverridecommandlockouts
\usepackage{cite}
\usepackage{amsmath,amssymb,amsfonts}
\usepackage{algorithmic}
\usepackage{graphicx}
\usepackage{textcomp}
\usepackage{xcolor}
\usepackage{cuted}
\usepackage{lipsum}
\usepackage{mathtools}
\usepackage{stfloats}
\usepackage{bbm}
\usepackage{multirow}

\newcommand{\R}{\mathbb{R}}
\newcommand{\N}{\mathbb{N}}
\newcommand{\T}{\top}
\newcommand{\I}{\mathbf{I}}
\newcommand{\0}{\mathbf{0}}

\newcommand{\diag}{\text{diag}}
\newcommand{\tsup}[1]{\textsuperscript{#1}}

\newcommand{\bm}[1]{\begin{bmatrix}#1\end{bmatrix}}
\def\BibTeX{{\rm B\kern-.05em{\sc i\kern-.025em b}\kern-.08em
    T\kern-.1667em\lower.7ex\hbox{E}\kern-.125emX}}

\newtheorem{assumption}{\textbf{Assumption}}
\newtheorem{theorem}{\textbf{Theorem}}
\newtheorem{corollary}{\textbf{Corollary}}
\newtheorem{lemma}{\textbf{Lemma}}
\newtheorem{proposition}{\textbf{Proposition}}
\newtheorem{remark}{\textbf{Remark}}
\newtheorem{definition}{\textbf{Definition}}

\usepackage{hyperref}

\begin{document}

\title{
Dissipativity-Based Distributed Droop-Free Controller and Communication Topology Co-Design for DC Microgrids
\\
}


\author{Mohammad Javad Najafirad and Shirantha Welikala 
\thanks{The authors are with the Department of Electrical and Computer Engineering, School of Engineering and Science, Stevens Institute of Technology, Hoboken, NJ 07030, \texttt{{\small \{mnajafir,swelikal\}@stevens.edu}}.}}

\maketitle

\pagenumbering{arabic}
\thispagestyle{plain}
\pagestyle{plain}

\begin{abstract}
This paper presents a novel dissipativity-based distributed droop-free control approach for the voltage regulation problem in DC microgrids (MGs) comprised of an interconnected set of distributed generators (DGs), loads, and power lines. First, we describe the closed-loop DC MG as a networked system where the sets of DGs and lines (i.e., subsystems) are interconnected via a static interconnection matrix. This interconnection matrix demonstrates how the inputs and outputs of DGs and lines are connected with each other. Each DG has a local controller and a distributed global controller. To design the distributed global controllers, we use the dissipativity properties of the subsystems and formulate a linear matrix inequality (LMI) problem. To support the feasibility of this distributed global controller design, we identify a set of necessary local conditions, which we then enforce in a specifically developed LMI-based local controller design process. In contrast to existing DC MG control solutions that separate distributed controller and communication topology design problems, our approach proposes a unified framework for distributed controller and communication topology co-design. As the co-design process is LMI-based, it can be efficiently implemented and evaluated using existing software tools. The effectiveness of the proposed solution in terms of voltage regulation and current sharing is verified by simulating an islanded DC MG in a MATLAB/Simulink environment under different scenarios, such as load changes and topological constraint changes, and comparing its performance with the recent droop control approach. 

\end{abstract}


\section{Introduction}







Environmental and economic concerns have driven the widespread integration of renewable energy sources (RESs) such as wind, photovoltaic panels, and fuel cells into power systems. Microgrids (MGs) provide a flexible framework to interconnect distributed generators (DGs), loads, energy storage systems (ESSs), control systems, and communication protocols, operating either in grid-connected or islanded modes. As a result, MGs are viewed as cyber-physical systems due to the close interaction between electrical and communication networks \cite{wang2019cyber}. Compared with AC MGs, DC MGs are gaining popularity due to the increasing demand for DC loads, such as in data centers, electric vehicles, and green buildings that reduce energy conversion losses by eliminating the need for AC/DC converters \cite{nasirian}.

The critical control challenges in DC MGs include voltage regulation (to some prespecified reference values) and current sharing (power sharing) among DGs. Numerous control strategies have been proposed to address these challenges, including centralized \cite{mehdi2020robust}, decentralized \cite{khorsandi2014decentralized}, and distributed control \cite{nasirian}. Centralized control achieves precise, flexible, and reliable control performance, but it suffers from a single point of failure and plug-and-play (PnP) capability \cite{guerrero2010hierarchical}. Decentralized control avoids the need for global information but often results in optimal power-sharing accuracy. To address these issues, distributed control is proposed, where DGs exchange local information through a sparse communication network to enhance the reliability, flexibility, and scalability of DC MGs \cite{dehkordi2016distributed}. 

One of the most widely adopted strategies for decentralized control in DC MGs is droop control. This method adjusts the voltage of each DG in proportion to its output current, mimicking the behavior of traditional synchronous machines in AC systems. Droop control facilitates proportional power sharing among DGs without requiring direct communication. In \cite{guerrero2010hierarchical}, a comprehensive hierarchical droop-based control structure for AC and DC MGs was introduced, showing how primary droop control can be augmented with secondary layers to restore voltage deviations. Another detailed analysis of droop control in AC and DC MGs is presented in \cite{palizban2015hierarchical}, emphasizing the necessity of secondary control to compensate for voltage and frequency deviations. 

However, droop control has its limitations. One of the primary drawbacks is voltage deviation, which occurs due to the trade-off between voltage regulation and power sharing. In low-voltage DC MGs, this issue is exacerbated by differences in line impedances. For instance, \cite{guerrero2010hierarchical} demonstrates that unequal line impedances can result in significant power-sharing inaccuracies. Although modified droop control methods, such as those proposed in \cite{zhou2020distributed}, attempt to address this issue by incorporating complex communication schemes, they still face system complexity and implementation cost challenges. Similarly, \cite{nasirian2014distributed} proposes integrating secondary control to improve voltage regulation and load-sharing accuracy, but at the expense of increased communication and system overhead. To solve this issue, authors in \cite{najafirad1} proposed a novel event-triggered algorithm to significantly reduce DG communication links. However, droop control requires precise tuning of droop coefficients to balance the trade-off between voltage regulation and power-sharing accuracy, which can be difficult in dynamic environments with varying loads. This issue is further explored in \cite{shuai2016microgrid}, where system stability degrades significantly under significant load variations, prompting the need for advanced communication-based secondary control. 

Given the inherent limitations of droop control, particularly in terms of voltage regulation and sensitivity to line impedance mismatches, researchers have begun exploring droop-free control strategies for DC MGs. One promising solution is the dissipativity-based control framework, which eliminates the need for droop characteristics by focusing on energy interactions between DGs, loads, and transmission lines. Unlike droop control, this approach ensures precise voltage regulation and proportional power sharing through energy-based control of the dynamics, making it more smooth and robust to impedance mismatches.

The proposed dissipativity-based droop-free control framework leverages the concept of dissipativity from system theory to model DGs, loads, and lines as interconnected energy systems. This control strategy ensures stability and optimal power flow across the DC MG. The problem of co-designing both distributed controllers and communication topology is formulated as a linear matrix inequality (LMI) problem, allowing for an efficient, decentralized and compositional implementation.

In all of the previously mentioned MG control strategies, the controller synthesis is typically separated from the communication topology, with the latter often assumed to be fixed or predefined. However, in real-world applications, the communication topology may vary due to the inherent intermittent nature of DGs, such as load changes or the plug-and-play (PnP) nature of DGs being added to or removed from the MG. Moreover, with advancements in communication technologies, such as software-defined radio networks, designing fixed communication topologies is no longer necessary. This flexibility opens the door to new control strategies adapting to dynamic network configurations. Thus, we are motivated to propose a co-design approach that simultaneously designs both the distributed controllers and the communication topology of the DC MG. Considering both aspects, we guarantee robust performance with an optimized communication topology (without fully connected communication topology) among all DGs. 

The main contributions of this paper can be summarized as follows:
\begin{enumerate}
\item
We formulate the DC MG control problem as a networked system control problem and propose the first (to the best of our knowledge) control theoretic co-design technique to jointly optimize the DC MG distributed controllers and the communication topology.
\item
We use a dissipativity-based distributed droop-free control approach that improves voltage regulation and current sharing, addressing the limitations of droop control while enhancing system stability.
\item 
We formulate the co-design problem as an LMI-based convex optimization problem, enabling efficient and scalable solutions as well as possible decentralized implementations.
\item 
To support the feasibility of co-design process, we introduce local controllers along with a systematic LMI-based design technique for such local controllers.
\end{enumerate}

The structure of this paper is as follows. Section \ref{Preliminaries} provides an overview of several essential notations and foundational concepts related to dissipativity and networked systems. Section \ref{problemformulation} discusses the DC MG modeling and the control problem formulation. In Section \ref{Passivity-based Control}, the proposed dissipativity-based control and topology co-design technique is introduced. The effectiveness of the proposed technique is demonstrated through simulations with comparisons in Section \ref{Simulation} before concluding the paper in Section \ref{Conclusion}.

\section{Preliminaries}\label{Preliminaries}
\subsection{Notations}
The notation $\mathbb{R}$ and $\mathbb{N}$ signify the sets of real and natural numbers, respectively. 
For any $N\in\mathbb{N}$, we define $\mathbb{N}_N\triangleq\{1,2,..,N\}$.
An $n \times m$ block matrix $A$ is denoted as $A = [A_{ij}]_{i \in \mathbb{N}_n, j \in \mathbb{N}_m}$, where $A_{ij}$ represents the $(i,j)^{\text{th}}$ block of $A$. Either subscripts or superscripts are used for indexing purposes, e.g., $A_{ij} = A^{ij}$.
$[A_{ij}]_{j\in\mathbb{N}_m}$ and $\diag([A_{ii}]_{i\in\mathbb{N}_n})$ represent a block row matrix and a block diagonal matrix, respectively.
$A^\top$ represents the transpose of a matrix $A$, and $(A^\top)^{-1}=A^{-\top}$.
$\0$ and $\I$, respectively, are the zero and identity matrices (their dimensions will be clear from the context). 
The representation of a symmetric positive definite (semi-definite) matrix $A\in\mathbb{R}^{n\times n}$ is $A>0\ (A\geq0)$. The symbol $\star$ represents conjugate matrices inside symmetric matrices. The Hermitian part of a matrix $A$ is defined as $\mathcal{H}(A)\triangleq A + A^\T$.
\subsection{Dissipativity}
Consider a dynamic control system
\begin{equation}\label{dynamic}
\begin{aligned}
    \dot{x}(t)=f(x(t),u(t)),\\
    y(t)=h(x(t),u(t)),
    \end{aligned}
\end{equation}
where $x(t)\in\mathbb{R}^n$, $u(t)\in\mathbb{R}^q$, $y(t)\in\mathbb{R}^m$, and $f:\mathbb{R}^n\times\mathbb{R}^q\rightarrow\mathbb{R}^n$ and $h:\mathbb{R}^n\times\mathbb{R}^q\rightarrow\mathbb{R}^m$ are continuously differentiable. 
The equilibrium points of (\ref{dynamic}) are such that there is a unique $u^*\in\mathbb{R}^q$ such that $f(x^*,u^*)=0$ for any $x^*\in\mathcal{X}$, where $\mathcal{X}\subset\mathbb{R}^n$ is the set of equilibrium states. And both $u^*$ and $y^*\triangleq h(x^*,u^*)$ are implicit functions of $x^*$.

The \textit{equilibrium-independent-dissipativity} (EID) \cite{arcak2022} is defined next to examine dissipativity of (\ref{dynamic}) without knowing exactly where its equilibrium points are. 

\begin{definition}
The system \eqref{dynamic} is called EID under supply rate $s:\mathbb{R}^q\times\mathbb{R}^m\rightarrow\mathbb{R}$ if there is a continuously differentiable storage function $V:\mathbb{R}^n\times\mathcal{X}\rightarrow\mathbb{R}$ such that  $V(x,x^*)>0$ when $x\neq x^*$, $V(x^*,x^*)=0$, and $\dot{V}(x,x^*)=\nabla_xV(x,x^*)f(x,u)\leq s(u-u^*,y-y^*)$, for all $(x,x^*,u)\in\mathbb{R}^n\times\mathcal{X}\times\mathbb{R}^q$.
\end{definition}

This $\textit{EID}$ property may be specialized depending on the utilized supply rate $s(.,.)$.
The concept of the $\textit{X-EID}$ property is defined in the sequel. This property is characterized by a quadratic supply rate, which is determined by its coefficient matrix $X=X^\top\in\mathbb{R}^{q+m}$ \cite{WelikalaP32022}.

\begin{definition}
The system (\ref{dynamic}) is $\textit{X-EID}$, where $X\triangleq[X^{kl}]_{k,l\in\mathbb{N}_2}$, if it is $\textit{EID}$ under the quadratic supply rate:
\begin{equation*}
    s(u-u^*,y-y^*)\triangleq\begin{bmatrix}
        u-u^* \\ y-y^*
    \end{bmatrix}^\top\begin{bmatrix}
        X^{11} & X^{12}\\ X^{21} & X^{22}
    \end{bmatrix}\begin{bmatrix}
        u-u^* \\ y-y^*
    \end{bmatrix}. 
\end{equation*}
\end{definition}

\begin{remark}\label{Rm:X-DissipativityVersions}
If the system (\ref{dynamic}) is $\textit{X-EID}$ with:\\
1)\ $X = \begin{bmatrix}
    \0 & \frac{1}{2}\I \\ \frac{1}{2}\I & \0
\end{bmatrix}$, then it is passive;\\
2)\ $X = \begin{bmatrix}
    -\nu\I & \frac{1}{2}\I \\ \frac{1}{2}\I & -\rho\I
\end{bmatrix}$, then it is strictly passive ($\nu$ and $\rho$ are the input feedforward and output feedback passivity indices, appropriately denoted as IFP($\nu$), OFP($\rho$) or IF-OFP($\nu,\rho$));\\
3) \ $X = \begin{bmatrix}
    \gamma^2\I & \0 \\ \0 & -\I
\end{bmatrix}$, then it is $L_2$-stable ($\gamma$ is the $L_2$-gain, denoted as $L2G(\gamma)$);\\
in an equilibrium-independent manner. 
\end{remark}

If the system (\ref{dynamic}) is linear time-invariant (LTI), a necessary and sufficient condition for the system (\ref{dynamic}) to be $\textit{X-EID}$ is provided in the following proposition as a linear matrix inequality (LMI) problem.

\begin{proposition}\label{Prop:linear_X-EID} \cite{welikala2023platoon}
The LTI system
\begin{equation}\label{Eq:Prop:linear_X-EID_1}
\begin{aligned}
    \dot{x}(t)=&Ax(t)+Bu(t),\\
    y(t)=&Cx(t)+Du(t),
\end{aligned}
\end{equation}
is $\textit{X-EID}$ if and only if there exists $P>0$ such that
\begin{equation}\label{Eq:Prop:linear_X-EID_2}
    \begin{bmatrix}
        -\mathcal{H}(PA)+C^\top X^{22}C & -PB+C^\top X^{21}+C^\top X^{22}D\\
        \star & X^{11}+\mathcal{H}(X^{12}D)+D^\top X^{22}D
    \end{bmatrix}\geq0.
\end{equation}
\end{proposition}

The following corollary takes into account a specific LTI system that has a local controller (a setup that we will see later on) and formulates an LMI problem for synthesizing this local controller to enforce (and optimize) the system $\textit{X-EID}$ property.
\begin{corollary}\label{Col.LTI_LocalController_XEID}\cite{welikala2023platoon}
The LTI system 
\begin{equation}
    \dot{x}(t)=(A+BL)x(t)+\eta(t),
\end{equation}
is $\textit{X-EID}$ (from input $\eta(t)$ to state $x(t)$) with $X^{22}<0$ if and only if there exists $P>0$ and $K$ such that
\begin{equation}
    \begin{bmatrix}
        -(X^{22})^{-1} & P & 0\\
        \star & -\mathcal{H}(AP+BK) & -\I+PX^{21}\\
        \star & \star & X^{11}
    \end{bmatrix}\geq0,
\end{equation}
and $L=KP^{-1}$.
\end{corollary}

\subsection{Networked Systems}\label{SubSec:NetworkedSystemsPreliminaries}
\subsubsection{Configuration} The networked system $\Sigma$ depicted in Fig. \ref{Networked} is composed of independent dynamic subsystems $\Sigma_i,i\in\mathbb{N}_N$, $\Bar{\Sigma}_i,i\in\mathbb{N}_{\Bar{N}}$ and a static interconnection matrix $M$ that characterizes the interconnections between the subsystems, an exogenous input signal $w(t)\in\mathbb{R}^r$ (e.g. disturbance) and an interested output signal $z(t)\in\mathbb{R}^l$ (e.g. performance). 
\begin{figure}
    \centering
    \includegraphics[width=0.5\columnwidth]{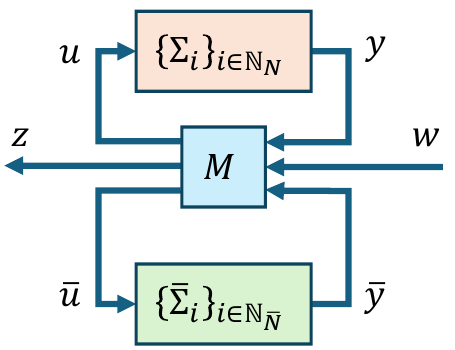}
    \caption{A generic networked system $\Sigma$.}
    \label{Networked}
\end{figure}

The dynamics of each subsystem $\Sigma_i,i\in\mathbb{N}_N$ are given by
\begin{equation}
    \begin{aligned}
        \dot{x}_i(t)&=f_i(x_i(t),u_i(t)),\\
        y_i(t)&=h_i(x_i(t),u_i(t)),
    \end{aligned}
\end{equation}
where $x_i(t)\in\mathbb{R}^{n_i}$, $u_i(t)\in\mathbb{R}^{q_i}$, $y_i(t)\in\mathbb{R}^{m_i}$. Similar to \eqref{dynamic}, each subsystem $\Sigma_i, i\in\N_N$ is considered to have a set $\mathcal{X}_i \subset \R^{n_i}$, where for every $x_i^* \in \mathcal{X}_i$, there exists a unique $u_i^* \in \R^{q_i}$ such that $f_i(x_i^*,u_i^*)=0$, and both $u_i^*$ and $y_i^*\triangleq h_i(x_i^*,u_i^*)$ are implicit function of $x_i^*$. Moreover, each subsystem $\Sigma_i, i\in\N_N$ is assumed to be $X_i$-EID, where $X_i \triangleq [X_i^{kl}]_{k,l\in\N_2}$ (known). Along the same lines, regarding each subsystem $\bar{\Sigma}_i, i\in\N_{\bar{N}}$, we use similar notations and assumptions. However, to distinguish among two types of subsystems, we include an extra bar symbol on corresponding notations to the subsystems $\bar{\Sigma}_i, i\in\N_{\bar{N}}$, e.g.,  $\bar{\Sigma}_i, i\in\N_{\bar{N}}$ is assumed to be $\bar{X}_i$-EID where $\bar{X}_i \triangleq [\bar{X}_i^{kl}]_{k,l\in\N_2}$.

On the other hand, the interconnection matrix $M$ and the corresponding interconnection relationship are given by
\begin{equation}\label{interconnectionMatrix}
    \begin{bmatrix}
        u \\ \bar{u} \\ z
    \end{bmatrix}=M\begin{bmatrix}
        y \\ \bar{y} \\ w
    \end{bmatrix}\equiv\begin{bmatrix}
        M_{uy} & M_{u\bar{y}} & M_{uw}\\
        M_{\bar{u}y} & M_{\bar{u}\bar{y}} & M_{\bar{u}w}\\
        M_{zy} & M_{z\bar{y}} & M_{zw}
    \end{bmatrix}\begin{bmatrix}
        y \\ \bar{y} \\ w
    \end{bmatrix},
\end{equation}
where $u\triangleq[u_i^\top]^\top_{i\in\mathbb{N}_N}$, $y\triangleq[y_i^\top]^\top_{i\in\mathbb{N}_N}$ and $\Bar{u}\triangleq[\Bar{u}_i^\top]^\top_{i\in\mathbb{N}_{\Bar{N}}}$, $\Bar{y}\triangleq[y_i^\top]^\top_{i\in\mathbb{N}_{\Bar{N}}}$.

\subsubsection{Dissipativity Analysis}
Inspired by \cite{arcak2022}, the following proposition exploits the $X_i$-EID and $\bar{X}_i$-EID properties of the respective individual subsystems $\Sigma_i,i\in\mathbb{N}_N$ and $\Bar{\Sigma}_i,i\in\mathbb{N}_{\Bar{N}}$ to formulate an LMI problem to analyze the \textbf{Y}-EID property of the networked system $\Sigma$, where $\text{\textbf{Y}}=\text{\textbf{Y}}^\top =  [\text{\textbf{Y}}^{kl}]_{k,l\in\mathbb{N}_2}$ is prespecified.

\begin{remark}\label{L2stable}
A good candidate for \textbf{Y} in \textbf{Y}-EID property is 
\begin{equation}\label{Eq:YEID_Candidate}
    \mathbf{Y}=\begin{bmatrix}
        \gamma^2\mathbf{I} & 0 \\ 0 & -\mathbf{I}
    \end{bmatrix}=\begin{bmatrix}
        \Tilde{\gamma}\I & 0 \\ 0 & -\I
    \end{bmatrix},
\end{equation}
which ensures the $\mathcal{L}_2$ stability in the networked system $\Sigma$. 
\end{remark}

\begin{proposition}
   \cite{welikala2023platoon} Let $\hat{x}^*\triangleq[[x_i^*{^\top}]_{i\in\mathbb{N}_N},[\bar{x}_i^{*\top}]_{i\in\mathbb{N}_{\bar{N}}}]^\top$ be an equilibrium point of $\Sigma$ (under $w(t)=0$) with $x_i^*\in\mathcal{X}_i,\forall i\in\mathbb{N}_N$, $\bar{x}_i^*\in\bar{\mathcal{X}}_i,\forall i\in\mathbb{N}_{\bar{N}}$ and $z^*$ be corresponding output. Then, the network is \textbf{Y}-EID if there exist scalars $p_i\geq0,\forall i\in\mathbb{N}_N$ and $\bar{p}_i\geq0,\forall i\in\mathbb{N}_{\bar{N}}$ such that
\begin{equation}
    \begin{bmatrix}
        M_{uy} & M_{u\bar{y}} & M_{uw}\\
        \I & \0 & \0 \\
        M_{\bar{u}y} & M_{\bar{u}\bar{y}} & M_{\bar{u}w}\\
        \0 & \I & \0\\
        \0 & \0 & \I\\
        M_{zy} & M_{z\bar{y}} & M_{zw}
    \end{bmatrix}^\top\begin{bmatrix}
        \textbf{X}_p & \0 & \0\\
        \0 & \bar{\textbf{X}}_{\bar{p}} & \0\\
        \0 & \0  & -\textbf{Y}
    \end{bmatrix}\star\leq0,
\end{equation}
where $\textbf{X}_p = [\textbf{X}_p^{kl}]_{k,l\in\N_2}$ with 
$\textbf{X}_p^{kl} = \diag([p_iX_i^{kl}]_{i\in\N_N})$ and
$\Bar{\textbf{X}}_{\bar{p}} = [\bar{\textbf{X}}_{\bar{p}}^{kl}]_{k,l\in\N_2}$ with 
$\bar{\textbf{X}}_{\bar{p}}^{kl} = \diag([\bar{p}_i\bar{X}_i^{kl}]_{i\in\N_{\bar{N}}})$. 
\end{proposition}

\subsubsection{Topology Synthesis}
The following proposition develops an LMI problem to synthesize the interconnection matrix $M$ (\ref{interconnectionMatrix}) to enforce the $\textbf{Y}$-EID property for the networked system $\Sigma$. Nonetheless, similar to \cite{welikala2023non}, we must initially establish two mild assumptions.

\begin{assumption}\label{As:NegativeDissipativity}
    For the networked system $\Sigma$, the provided \textbf{Y}-EID specification is such that $\textbf{Y}^{22}<0$.
\end{assumption}
\begin{remark}
According to Rm. \ref{Rm:X-DissipativityVersions}, As. \ref{As:NegativeDissipativity} holds if the networked system $\Sigma$ must be either: (i) L2G($\gamma$) or (ii) OFP($\rho$) (or IF-OFP($\nu,\rho$)) with some $\rho>0$, i.e., $L_2$-stable or passive, respectively. Therefore, As. \ref{As:NegativeDissipativity} is mild since it is usually preferable to make the networked system $\Sigma$ either $L_2$-stable or passive.
\end{remark}

\begin{assumption}\label{As:PositiveDissipativity}
    Every subsystem $\Sigma_i$ in the networked system $\Sigma$ is $X_i$-EID with $X_i^{11}>0, \forall i\in\mathbb{N}_N$. Similarly, every subsystem $\bar{\Sigma}_i$ 
    in the networked system $\Sigma$ 
    is $\bar{X}_i$-EID with $\bar{X}_i^{11}>0, \forall i\in\N_{\bar{N}}$.
\end{assumption}

\begin{remark}
    According to Rm. \ref{Rm:X-DissipativityVersions}, As. \ref{As:PositiveDissipativity} holds if a subsystem $\Sigma_i,i\in\N_N$ is either: (i) L2G($\gamma_i$) or (ii) IFP($\nu_i$) (or IF-OFP($\nu_i,\rho_i$)) with $\nu_i<0$ (i.e., $L_2$-stable or non-passive). Since in passivity-based control, often the involved subsystems are non-passive (or can be treated as such even if they are passive), As. \ref{As:PositiveDissipativity} is also mild. 
\end{remark}

\begin{proposition}\label{synthesizeM}\cite{welikala2023platoon}
    Under Assumptions 1 and 2, the network system in Fig. \ref{Networked} can be made \textbf{Y}-EID (from $w(t)$ to output $z(t)$) by synthesizing the interconnection matrix $M$ via solving LMI problem:
\begin{equation}
\begin{aligned}
	&\mbox{Find: } 
	L_{uy}, L_{u\bar{y}}, L_{uw}, L_{\bar{u}y}, L_{\bar{u}\bar{y}}, L_{\bar{u}w}, M_{zy}, M_{z\bar{y}}, M_{zw}, \\
	&\mbox{Sub. to: } p_i \geq 0, \forall i\in\N_N, \ \ 
	\bar{p}_l \geq 0, \forall l\in\N_{\bar{N}},\ \text{and} \  \eqref{NSC4YEID},
\end{aligned}
\end{equation}
with
\scriptsize
$\bm{M_{uy} & M_{u\bar{y}} & M_{uw} \\ M_{\bar{u}y} & M_{\bar{u}\bar{y}} & M_{\bar{u}w}} = 
\bm{\textbf{X}_p^{11} & \0 \\ \0 & \bar{\textbf{X}}_{\bar{p}}^{11}}^{-1} \hspace{-1mm} \bm{L_{uy} & L_{u\bar{y}} & L_{uw} \\ L_{\bar{u}y} & L_{\bar{u}\bar{y}} & L_{\bar{u}w}}$.
\normalsize
\end{proposition}

\begin{figure*}[!hb]
\centering
\hrulefill
\begin{equation}\label{NSC4YEID}
\scriptsize
 \bm{
		\textbf{X}_p^{11} & \0 & \0 & L_{uy} & L_{u\bar{y}} & L_{uw} \\
		\0 & \bar{\textbf{X}}_{\bar{p}}^{11} & \0 & L_{\bar{u}y} & L_{\bar{u}\bar{y}} & L_{\bar{u}w}\\
		\0 & \0 & -\textbf{Y}^{22} & -\textbf{Y}^{22} M_{zy} & -\textbf{Y}^{22} M_{z\bar{y}} & \textbf{Y}^{22} M_{zw}\\
		L_{uy}^\T & L_{\bar{u}y}^\T & - M_{zy}^\T\textbf{Y}^{22} & -L_{uy}^\T\textbf{X}^{12}-\textbf{X}^{21}L_{uy}-\textbf{X}_p^{22} & -\textbf{X}^{21}L_{u\bar{y}}-L_{\bar{u}y}^\T \bar{\textbf{X}}^{12} & -\textbf{X}^{21}L_{uw} + M_{zy}^\T \textbf{Y}^{21} \\
		L_{u\bar{y}}^\T & L_{\bar{u}\bar{y}}^\T & - M_{z\bar{y}}^\T\textbf{Y}^{22} & -L_{u\bar{y}}^\T\textbf{X}^{12}-\bar{\textbf{X}}^{21}L_{\bar{u}y} & 		-(L_{\bar{u}\bar{y}}^\T \bar{\textbf{X}}^{12} + \bar{\textbf{X}}^{21}L_{\bar{u}\bar{y}}+\bar{\textbf{X}}_{\bar{p}}^{22}) & -\bar{\textbf{X}}^{21} L_{\bar{u}w} + M_{z\bar{y}}^\T \textbf{Y}^{21} \\ 
		L_{uw}^\T & L_{\bar{u}w}^\T & -M_{zw}^\T \textbf{Y}^{22}& -L_{uw}^\T\textbf{X}^{12}+\textbf{Y}^{12}M_{zy} & -L_{\bar{u}w}^\T\bar{\textbf{X}}^{12}+ \textbf{Y}^{12} M_{z\bar{y}} & M_{zw}^\T\textbf{Y}^{21} + \textbf{Y}^{12}M_{zw} + \textbf{Y}^{11}
	}>0 
\end{equation}
\end{figure*}

\subsection{Graph Theory}\label{graph theory}
In this subsection, we recall some fundamental graph theoretic definitions that will be useful to describe the physical and communication topologies of a DC MG. 

We denote by $\mathcal{G}=(\mathcal{V},\mathcal{E})$ a directed graph, where the $\mathcal{V}=\{v_1,v_2,...,v_n\}$ represents the set of nodes and $\mathcal{E} = \{e_1,e_2,...,e_m\} \subseteq(\mathcal{V}\times\mathcal{V})$ represents the set of edges. For a given node $v_i\in\mathcal{V}$, the set of out-neighbors is $\mathcal{E}_i^+=\{v_j \in\mathcal{V}: (v_i,v_j)\in\mathcal{E}\}$, the set of in-neighbors is $\mathcal{E}_i^-=\{v_j\in\mathcal{V}: (v_j,v_i)\in\mathcal{E}\}$, and the set of neighbors is $\mathcal{E}_i=\mathcal{E}_i^+\cup\mathcal{E}_i^-$.
The incidence matrix $\mathcal{B} = [\mathcal{B}_{il}]_{i\in\N_n,l\in\N_m}\in \R^{n\times m}$ of this graph $\mathcal{G}$ represents the connectivity between its nodes $\mathcal{V}=\{v_1,v_2,...,v_n\}$ and edges $\mathcal{E} = \{e_1,e_2,...,e_m\}$, where   
\begin{equation}
    \mathcal{B}_{il} = 
    \begin{cases}
      +1 & \mbox{ if edge } e_l \mbox{ is directed onto the node } v_i,\\
      -1 & \mbox{ if edge } e_l \mbox{ is directed outward from node } v_i,\\
      0 & \mbox{otherwise}.
    \end{cases}
\end{equation}
The adjacency matrix $\mathcal{A} = [\mathcal{A}_{ij}]_{i\in\N_n,j\in\N_n} \in \R^{n\times n}$ of this graph represents the connectivity between its nodes $\mathcal{V}=\{v_1,v_2,...,v_n\}$, where   
\begin{equation}
    \mathcal{A}_{ij} = 
    \begin{cases}
      +1 & \mbox{ if } v_j \in \mathcal{E}_i^+,\\
      -1 & \mbox{ if } v_j \in \mathcal{E}_i^-,\\
      0 & \mbox{otherwise}.
    \end{cases}
\end{equation}




\section{Problem Formulation}\label{problemformulation}
This section presents the dynamic modeling details of the DC MG, consisting of multiple DGs and loads interconnected via transmission lines. Specifically, we adopt the model proposed in \cite{nahata}, which enables studying the role and the impact of communication and physical topologies in DC MGs given the dynamic models of DG units, loads and transmission lines. We also introduce local and global controllers for the DC MG. Finally, we derive the closed-loop networked system representation of the DC MG.

\subsection{DC MG Physical Interconnection Topology}
The physical interconnection topology of a DC MG can be represented by a directed connected graph $\mathcal{G}^p =(\mathcal{V},\mathcal{E})$ where $\mathcal{V}$ is divided into two parts (i.e., bipartite): $\mathcal{D}=\{\Sigma_i^{DG}, i\in\N_N\}$ represents the set of DGs, and $\mathcal{L}=\{\Sigma_l^{line}, l\in\N_L\}$ represents the set of transmission lines. The DGs are interconnected with each other through transmission lines. As each transmission line $\Sigma_l^{line}$ connects only two DGs, let us denote $\Sigma_l^{line} \equiv (\Sigma_j^{DG},\Sigma_k^{DG})$ where DG indices $j$ and $k$ will be determined by the line index $l$.
The interface between each DG and the DC MG is established through a point of common coupling (PCC). For simplicity, the loads are assumed to be connected to the DG terminals at the respective PCCs. Indeed, the loads can be moved to PCCs using Kron reduction even if they are in a different location \cite{dorfler2012kron}. 
Given that every DG is only directly connected to the lines, every edge in $\mathcal{E}$ has one node in $\mathcal{D}$ and another node in $\mathcal{L}$. This results in $\mathcal{G}^p$ being a directed bipartite graph. The orientation of each edge reflects an arbitrarily chosen reference direction for positive currents. As mentioned before, the line $\Sigma_l^{line}\equiv (\Sigma_j^{DG}, \Sigma_k^{DG}), l\in\N_L$ has an in-neighbor $\Sigma_j^{DG}$ and an out-neighbor $\Sigma_k^{DG}$ since the current entering a line must exit it. 

To represent the DC MG's physical topology, we use its adjacency matrix $\mathcal{A}\in\R^{(N+L) \times (N+L)}$ that describes the connectivity between its nodes (recall that a node may be a DG or a line), where 
\begin{equation*}
   \mathcal{A} = \begin{bmatrix}
    \0_{N \times N} & \mathcal{B} \\
    \mathcal{B}^\T & \0_{L\times L} 
\end{bmatrix} \in \R^{(N+L) \times (N+L)},
\end{equation*} 
and $\mathcal{B}\in\R^{N \times L}$ is the incident matrix of the DG network (where nodes are just the DGs and edges are just the transmission lines). Note that $\mathcal{B}$ can also be viewed as a special adjacency matrix of $\mathcal{G}^p$ that describes the connectivity between its two types of nodes. Therefore, $\mathcal{B}$ is also called the ``bi-adjacency'' matrix of $\mathcal{G}^p$.
In particular, the bi-adjacency matrix $\mathcal{B} \in \R^{N\times L}$ is defined as $\mathcal{B}=[\mathcal{B}_{il}]_{i \in \N_N, l \in \N_L}$ where 
\begin{equation}\label{Eq:BiAdjacencyMat}
\mathcal{B}_{il}=
    \begin{cases}
      +1 & \text{if } l\in\mathcal{E}_i^+, \\
      -1 & \text{if } l\in\mathcal{E}_i^-, \\
      0 & \text{otherwise}. 
    \end{cases}
\end{equation}



\subsection{Dynamic Model of a Distributed Generator (DG)}
Each DG consists of a DC voltage source, a voltage source converter (VSC), and some RLC components. The DG $\Sigma_i^{DG},i\in\N_N$ supplies power to a specific local load at its PCC (denoted $\text{PCC}_i$). Additionally, it interconnects with other DG units via transmission lines $\{\Sigma_l^{line}:l \in \mathcal{E}_i\}$. Fig. \ref{DCMG} illustrates the schematic diagram of $\Sigma_i^{DG}$, including the load, a connected transmission line, and the local and distributed global controllers.

By applying Kirchhoff's Current Law (KCL) and Kirchhoff's Voltage Law (KVL) at $\text{PCC}_i$ on the DG side, we get the following equations for $\Sigma_i^{DG}$:
\begin{equation}
\begin{aligned}\label{DGEQ}
C_{ti}\frac{dV_i}{dt} &= I_{ti} - I_{Li}(V_i) - I_i, \\
L_{ti}\frac{dI_{ti}}{dt} &= -V_i - R_{ti}I_{ti} + V_{ti},
\end{aligned}
\end{equation}
where the parameters $R_{ti}$, $L_{ti}$, and $C_{ti}$ 
represent the internal resistance, internal inductance, and filter capacitance of $\Sigma_i^{DG}$, respectively. The state variables are selected as $V_i$ and $I_{ti}$, where $V_i$ is the $\text{PCC}_i$ voltage and $I_{ti}$ is the internal current. Moreover, $V_{ti}$ is the input command signal applied to the VSC, $I_{Li}(V_i)$ is the load current, and $I_i$ is the total current injected to the DC MG by $\Sigma_i^{DG}$. 
Note that $V_{ti}$, $I_{Li}(V_i)$, and $I_i$ are respectively determined by the controllers, loads, and lines at $\Sigma_i^{DG}$. For example, $I_i$ is given by 
\begin{equation}
\label{Eq:DGCurrentNetOut}
I_i=\sum_{l\in\mathcal{E}_i^+}\mathcal{B}_{il}I_l+\sum_{l\in\mathcal{E}_i^-}\mathcal{B}_{il}I_l=\sum_{l\in\mathcal{E}_i}\mathcal{B}_{il}I_l,
\end{equation}
where $I_l, l\in\mathcal{E}_i$ are line currents.

\subsection{Dynamic Model of a Transmission Line}
Each transmission line is modeled using the $\pi$-equivalent representation, where we assume that the line capacitances are consolidated with the capacitances of the DG filters. Consequently, as shown in Fig. \ref{DCMG}, the power line $\Sigma_l^{line}$ can be represented as an RL circuit with resistance $R_l$ and inductance $L_l$. By applying KVL to $\Sigma_l^{line}$, we obtain:
\begin{equation}\label{line}
    \Sigma_l^{line}:
        L_l\frac{dI_l}{dt}=-R_lI_l+\Bar{u}_l,
\end{equation}
where
\begin{equation*}
    \Bar{u}_l=V_i-V_j=\sum_{i\in \mathcal{E}_l}\mathcal{B}_{il}V_i.
\end{equation*}
\begin{figure}
    \centering
    \includegraphics[width=\columnwidth]{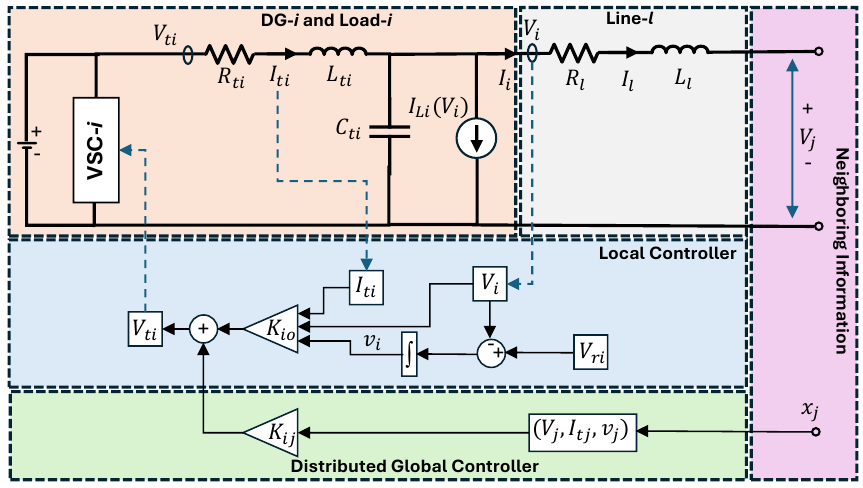}
    \caption{The electrical schematic of DG-$i$, load-$i$, $i\in\N_N$, local controller, distributed global controller, and line-$l$, $l\in\N_L$.}
    \label{DCMG}
\end{figure}

\subsection{Load Model} 
Recall that $I_{Li}(V_i)$ in Fig. \ref{DCMG} and \eqref{DGEQ} represents the current flowing through the load connected at $\Sigma_i^{DG}$. The exact form of $I_{Li}(V_i)$ depends on the type of load. In the most general case \cite{kundur2007power}, the load can be thought of as a ``ZIP'' load where $I_{Li}(V_i)$ takes the form: 
\begin{equation}\label{load}
I_{Li}(V_i) = I_{Li}^Z(V_i) + I_{Li}^I(V_i) + I_{Li}^P(V_i).
\end{equation}
Here, the ZIP load's components are: 
(i) a constant impedance load: $I_{Li}^{Z}(V_i)=Y_{Li}V_i$, where  $Y_{Li}=1/R_{Li}$ is the conductance of the load,
(ii) a constant current load: $I_{Li}^{I}(V_i)=\bar{I}_{Li}$ where $\bar{I}_{Li}$ is the current demand of the load, and
(iii) a constant power load: $I_{Li}^{P}(V_i)=V_i^{-1}P_{Li}$, where $P_{Li}$ represents the power demand of the load.

\subsection{Voltage Regulation and Current Sharing Formulation}




In this section, we formulate two common control objectives for DC MGs. First, by assuming $P_{Li}\triangleq0$, we get the dynamics of the DC MG by combining the DG and line dynamics in (\ref{DGEQ}) and (\ref{line}):
\begin{equation}\label{Compactformofdynamics}
\begin{aligned}
    C_t\dot{V} &= I_t-Y_{L}V-I_{L}-\mathcal{B}I,\\
    L_t\dot{I}_t &= -V - R_tI_t + u,\\
    L\dot{I} &= -RI - \mathcal{B}^\top V.
 \end{aligned}
\end{equation}
where $I_t,V,u\in\mathbb{R}^N$, and $I\in\mathbb{R}^L$ are vectorized forms of $I_{ti}$, $V_{i}$, $u_i$ and $I_{i}$, respectively. 
Moreover, $C_t,L_t,Y_L,I_L\in\mathbb{R}^{N\times N}$ and $R,L\in\mathbb{R}^{L\times L}$ are positive definite diagnoal matrices, e.g., $C_t=\diag\{C_{t1},...,C_{tN}\}$.
We notice that for a given constant input $\bar{u}$, a steady state solution $(\Bar{I}_t,\Bar{I},\Bar{V})$ to the above system (\ref{Compactformofdynamics}) satisfies:
\begin{equation}
    \begin{aligned}
        \Bar{I}_t &= \mathcal{B}\Bar{I} + Y_{L}\Bar{V} + I_L,\\
        \Bar{V} &=  R_t\Bar{I}_t - \Bar{u},\\
        \Bar{I} &= -R^{-1}\mathcal{B}^\T\Bar{V}.\\
    \end{aligned}
\end{equation}
This equation implies that at the steady state, the total generated current $\mathbbm{1}^\top \Bar{I}_t$ is proportional to the total current $\mathbbm{1}^\top(Y_L\Bar{V}+I_L)$ demanded by ZI load. To avoid stressing a DG unit, it is preferred that each DG supply the total demand based on its corresponding capacity ratio. This objective is equivalent to achieving $w_i\Bar{I}_{ti}=w_j\Bar{I}_{tj}$ for all $i,j\in\mathbb{N}_N$, where a $\Sigma_i^{DG}$ with a relatively large value of $w_i$ has a lesser contribution to supply the load demand. Following \cite{nasirian}, we aim to achieve average voltage regulation by selecting $1/w_i,\forall i\in\mathbb{N}$. We formulate the two control objectives for DC MGs as follows. \\
\textit{Objective 1}: (\textbf{Average voltage regulation})\\
\begin{equation}
     \lim_{t\to\infty}\mathbbm{1}^\top W^{-1}V(t)= \mathbbm{1}^\top W^{-1}\Bar{V} = \mathbbm{1}^\top W^{-1} V_r,
\end{equation}
\textit{Objective 2}: (\textbf{Current sharing})\\
\begin{equation}
    \lim_{t\to\infty}I_t(t)=\Bar{I}_t = W^{-1}\mathbbm{1}i_t^*,
\end{equation}
where $W=\diag(w_1,..,w_n), w_i>0$, for all $i\in\mathbb{N}_N$ and $i_t^*$ any scalar.

\subsection{The Local and Distributed Global Controllers}
The primary aim of the local and global controllers is to guarantee that the voltage at $\text{PCC}_i$ closely follows a predetermined reference voltage $V_{ri}(t)$ at each $\Sigma_i^{DG},i\in\N_N$. 
In the absence of voltage stabilization in the DC MG, there is a potential for voltages to exceed critical regulatory thresholds, leading to detrimental consequences for the loads as well as for the entire DC MG infrastructure. 

At $\Sigma_i^{DG}, i\in\N_N$, to effectively track the assigned reference voltage $V_{ri}(t)$, it is imperative to ensure that the error $e_i(t)\triangleq V_i(t)-V_{ri}(t)$ converges to zero. 
To achieve this objective, following the approach proposed in \cite{tucci2017}, we first include each $\Sigma_i^{DG}$ with an integrator state $v_i$ that follows the dynamics (see also Fig. \ref{DCMG})
\begin{equation}\label{error}
    \frac{dv_i}{dt}=e_i(t)=V_i(t)-V_{ri}(t).
\end{equation}
 Then, each $\Sigma_i^{DG}$ is equipped with a local state feedback controller
 \begin{equation}\label{Controller}
   u_{i0}(x_i(t)) 
   \triangleq K_{i0}x_i(t),
 \end{equation}
where $x_i \triangleq \begin{bmatrix}
    V_i &  I_{ti} & v_i
\end{bmatrix}^\top\in\mathbb{R}^3$ denotes the augmented state of $\Sigma_i^{DG}$ and $K_{i0}=\begin{bmatrix}
    k_{i0}^V & k_{i0}^I & k_{i0}^v
\end{bmatrix}\in\mathbb{R}^{1\times3}$ is the local state feedback controller gain matrix.
However, this local controller does not guarantee global stability in the presence of other interconnected DGs. To address this issue, as shown in Fig. \ref{DCMG}, a distributed global controller 
\begin{equation}
    \sum_{j\in\bar{\mathcal{F}}_i^-} u_{ij}(x_j(t)) = \sum_{j\in\bar{\mathcal{F}}_i^-} k_{ij}x_j(t),
\end{equation} 
is employed where each $k_{ij}=\begin{bmatrix}
    k_{ij}^V & k_{ij}^I & k_{ij}^v
\end{bmatrix}\in\mathbb{R}^{1\times3}$ is a distributed global state feedback controller gain matrix and $\bar{\mathcal{F}}_i^-$ is the communication-wise in-neighbors including $\{i\}$.
Thus, the overall controller $u_i(t)$ applied to the VSC of $\Sigma_i^{DG}$ can be expressed as
\begin{equation}\label{controlinput}
    u_i(t)\triangleq V_{ti}(t) =u_{i0}(x_i(t))+\sum_{j\in\bar{\mathcal{F}}_i^-} u_{ij}(x_j(t)).
\end{equation}

With the combination of the proposed local and distributed global controllers, the voltage of each $\Sigma_i^{DG}$ can track the corresponding reference voltage $V_{ri}(t)$. To guarantee voltage restoration, the distributed controller at $\Sigma_i^{DG}$ uses a sparse communication network so as to receive from and send to its communication-wise in- and out-neighbors, respectively.
In most studies, the communication topology of DC MG is predetermined. However, in this paper, we propose a novel approach to co-design the communication network and the distributed controllers. 

Note that we denote the communication topology as a directed graph $\mathcal{G}^c =(\mathcal{D},\mathcal{F})$ where $\mathcal{D}\triangleq\{\Sigma_i^{DG}, i\in\N_N\}$ (same as before) and $\mathcal{F}$ represents the set of communication links among DGs (to be designed). Note also that, for a given $\Sigma_i^{DG}$, the communication-wise out- and in-neighbors are denoted as $\mathcal{F}_i^+$ and $\mathcal{F}_i^-$, respectively. We also use the notations: $\bar{\mathcal{F}}_i^+ \triangleq \mathcal{F}_i^+ \cup \{i\}$ and $\bar{\mathcal{F}}_i^- \triangleq \mathcal{F}_i^- \cup \{i\}$ (see \eqref{controlinput}). 


\subsection{Closed-Loop Dynamics of the DC MG}

By combining the DG dynamics \eqref{DGEQ} and the error dynamics \eqref{error}, we can write the overall dynamics of $\Sigma_i^{DG}$ in the form:
\begin{align}\nonumber\label{statespacemodel}
       \frac{dV_i}{dt}&=\frac{1}{C_{ti}}I_{ti}-\frac{1}{C_{ti}}I_{Li}(V_i)-\frac{1}{C_{ti}}I_i,\\
        \frac{dI_{ti}}{dt}&=-\frac{1}{L_{ti}}V_i-\frac{R_{ti}}{L_{ti}}I_{ti}+\frac{1}{L_{ti}}u_i,\\\nonumber
        \frac{dv_i}{dt}&=V_i-V_{ri}.\nonumber
\end{align}
In \eqref{statespacemodel}, the terms $I_{Li}(V_i)$, $I_i$ and $u_i$ can be substituted with the expressions given in \eqref{load}, \eqref{Eq:DGCurrentNetOut} and  \eqref{controlinput}, respectively. Consequently, the dynamics of $\Sigma_i^{DG}$ can be stated as
\begin{equation}\label{statespacewithcontroller}
\begin{aligned}
    \begin{bmatrix}
        \dot{V}_i \\ \dot{I}_{ti} \\ \dot{v}_i
    \end{bmatrix}=&\begin{bmatrix}
        -\frac{Y_{Li}}{C_{ti}}V_i-\frac{P_{Li}}{C_{ti}}V_i^{-1}+\frac{1}{C_{ti}}I_{ti}\\
        -\frac{1}{L_{ti}}V_i - \frac{R_{ti}}{L_{ti}}I_{ti} + \frac{1}{L_{ti}}u_{i0}\\
        V_i
    \end{bmatrix}+\begin{bmatrix}
        -\frac{\Bar{I}_{Li}}{C_{ti}} \\ 0 \\ -V_{ri}
    \end{bmatrix}\\+&
    \begin{bmatrix}
        -\frac{1}{C_{ti}}\sum_{l\in\mathcal{E}_i} \mathcal{B}_{il}I_l \\ 0 \\ 0
    \end{bmatrix}+\begin{bmatrix}
        0 \\ \frac{1}{L_{ti}}\sum_{j\in\bar{\mathcal{F}}_i^-}u_{ij} \\ 0
    \end{bmatrix}.
\end{aligned}
\end{equation}
We can write \eqref{statespacewithcontroller} in the compact linear form (assuming there is no constant power load, i.e., $P_{Li}\triangleq0$, as specified in \eqref{load}): 
\begin{equation}
\label{Eq:DGCompact}
\begin{aligned}
\dot{x}_i(t)=& A_ix_i(t)+B_iu_i(t)+w_i(t)+\xi_i(t),\\
z_i(t)=& H_ix_i(t),
\end{aligned}
\end{equation}
where the exogenous input (disturbance) is defined as 
$w_i \triangleq \begin{bmatrix}
    -C_{ti}^{-1}\Bar{I}_{Li} & 0 & -V_{ri}
\end{bmatrix}^\top$,  
the transmission line coupling is defined as 
$\xi_i \triangleq \begin{bmatrix}
    -C_{ti}^{-1}\sum_{l\in \mathcal{E}_i} \mathcal{B}_{il}I_l & 0 & 0
\end{bmatrix}^\top$, and the desired performance metric (given that we want to ensure $\lim_{t\to\infty}V_i(t)=V_{ri}$) is defined as $z_i \triangleq v_i$. The system matrices $A_i$, $B_i$ and $H_i$ in \eqref{Eq:DGCompact} respectively are
\begin{equation}\label{Eq:DG_Matrix_definition}
A_i \triangleq
\begin{bmatrix}
  -\frac{Y_{Li}}{C_{ti}} & \frac{1}{C_{ti}} & 0\\
-\frac{1}{L_{ti}} & -\frac{R_{ti}}{L_{ti}} & 0 \\
1 & 0 & 0
\end{bmatrix}, 
B_i \triangleq
\begin{bmatrix}
 0 \\ \frac{1}{L_{ti}} \\ 0
\end{bmatrix}\ \mbox{ and }\ 
H_i \triangleq \bm{0 \\ 0 \\ 1}^\T.
\end{equation}

Similarly, using \eqref{line}, the state space representation of the transmission line $\Sigma_l^{Line}$ can be written in a compact form:
\begin{equation}\label{Eq:LineCompact}
    \dot{\bar{x}}_l(t) = \bar{A}_l\bar{x}_l(t) + \bar{B}_l\bar{u}_l,
 \end{equation}
where $\bar{x}_l \triangleq I_l$ is the transmission line state and  $\bar{u}_l \triangleq \sum_{i\in \mathcal{E}_l}\mathcal{B}_{il}V_i$ is the voltage difference across the transmission line.
The system matrices $\bar{A}_l$ and $\Bar{B}_l$ are defined respectively as
\begin{equation}\label{Eq:Line_Matrix_definition}
    \bar{A}_l \triangleq \begin{bmatrix}
        -\frac{R_l}{L_l} 
    \end{bmatrix}
    \quad \text{and} \quad 
    \Bar{B}_l \triangleq  \begin{bmatrix}
        \frac{1}{L_l}
    \end{bmatrix}.
\end{equation}

\subsection{Networked System Model}\label{Networked System Model}
Let us define (vectorize) 
$u\triangleq[u_i]_{i\in\N_N}$ and $\Bar{u}\triangleq[\Bar{u}_l]_{l\in\N_L}$ respectively as control inputs of DGs and lines, 
$x\triangleq[x_i]_{i\in\N_N}$ and $\Bar{x}\triangleq[\Bar{x}_l]_{l\in\N_L}$ respectively as the full state of DGs and lines, $w\triangleq[w_i]_{i\in\N_N}$ as disturbance inputs, and $z\triangleq[z_i]_{i\in\N_N}$ as performance outputs of the DC MG.

Using these vectorized quantities, we can represent the DC MG as two sets of subsystems (i.e., $\mathcal{D}=\{\Sigma_i^{DG}, i\in\N_N\}$ and $\mathcal{L}=\{\Sigma_l^{line}, l\in\N_L\}$) that are interconnected together through a static interconnection matrix ``$M$'' as shown in Fig. \ref{netwoked}. From comparing Fig. \ref{netwoked} with Fig. \ref{Networked}, it is clear that the DC MG takes the form of a standard networked system discussed in Sec. \ref{SubSec:NetworkedSystemsPreliminaries}. 

\begin{figure}
    \centering
    \includegraphics[width=0.9\columnwidth]{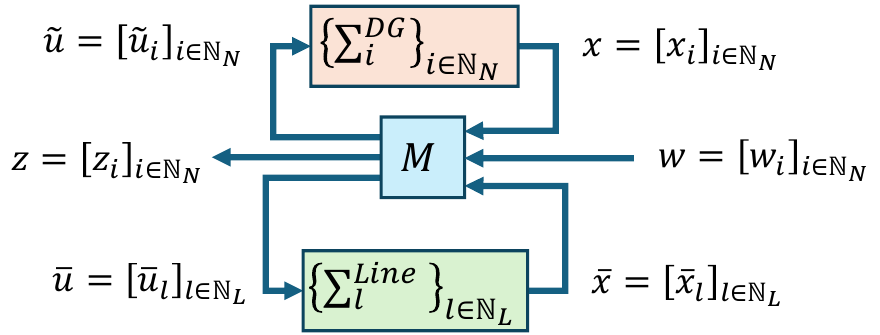}
    \caption{DC MG dynamics as a networked system configuration.}
    \label{netwoked}
\end{figure}

To identify the specific structure of the interconnection matrix $M$ in Fig. \ref{netwoked} (i.e., for DC MG), we need to closely observe how the dynamics of DGs and lines are interconnected with each other, and their dependence on disturbance inputs and their affects on performance outputs.  

To this end, let us use \eqref{Eq:DGCompact} and \eqref{controlinput} to state the closed-loop dynamics of $\Sigma_i^{DG}$ as 
\begin{equation}\label{closedloopdynamic}
    \dot{x}_i=(A_i+B_iK_{i0})x_i+\tilde{u}_i,
\end{equation}
where 
$\tilde{u}_i 
\triangleq w_i + \xi_i + B_i\sum_{j\in\bar{\mathcal{F}}_i^-}k_{ij}x_j(t)$. 
Note that $\tilde{u}_i$ in \eqref{closedloopdynamic} can be stated as 
\begin{equation}
\label{Eq:DGClosedLoopDynamics_varphi}
\tilde{u}_i =w_i+\sum_{l\in\mathcal{E}_i}\Bar{C}_{il}\Bar{x}_l+\sum_{j\in\bar{\mathcal{F}}_i^-}K_{ij}x_j,
\end{equation}
where we define 
$\Bar{C}_{il} \triangleq -C_{ti}^{-1}\begin{bmatrix}
        \mathcal{B}_{il} & 0 & 0
    \end{bmatrix}^\top, \forall l\in\mathcal{E}_i$, and 
\begin{equation}\label{k_ij}
    K_{ij} \triangleq \frac{1}{L_{ti}}\begin{bmatrix}
        0 & 0 & 0\\
        k_{ij}^V & k_{ij}^I & k_{ij}^v \\
        0 & 0 & 0
    \end{bmatrix}, \forall j\in\bar{\mathcal{F}}_i^-.
\end{equation}
By vectorizing \eqref{Eq:DGClosedLoopDynamics_varphi} over all $i\in\N_N$ and $l\in\N_L$, we get 
\begin{equation}\label{Eq:DGClosedLoopInputVector}
    \tilde{u} \triangleq [\tilde{u}_i]_{i\in \N_N} =  w+\Bar{C}\Bar{x}+Kx,
\end{equation}
where $\Bar{C}=[\Bar{C}_{il}]_{i\in\mathbb{N}_N,l\in\mathbb{N}_L}\in\R^{3N\times L}$ and $K=[K_{ij}]_{i,j\in\mathbb{N}_N} \in\R^{3N\times 3N}$.

\begin{remark}
The block matrices $K$ and $\Bar{C}$ in \eqref{Eq:DGClosedLoopInputVector} are indicative of the communication and physical topologies of the DC MG, respectively. In particular, the $(i,j)$\tsup{th} block in $K$, i.e., $K_{ij}$ indicates a communication link from $\Sigma_j^{DG}$ to $\Sigma_i^{DG}$. Similarly, $(i,l)$\tsup{th} block in $\bar{C}$ 
represents a physical link from $\Sigma_i^{DG}$ and $\Sigma_l^{Line}$ (see \eqref{Eq:BiAdjacencyMat}). Clearly, if we can design $K$ and $\Bar{C}$ in \eqref{Eq:DGClosedLoopInputVector}, we can deduce the communication and physical topologies of the DC MG. This implies that we can optimize the communication and physical topologies of the DC MG by carefully designing these matrices. However, it should also be noted that the structure, as well as some element values of these block matrices, may be predefined. For example, when the physical topology is predefined and fixed, so is the block matrix $\bar{C}$.    
\end{remark} 

In parallel to $\tilde{u}$, which represents the effective input vector to the DGs, we also need to identify $\bar{u}$, which represents the effective input vector to the lines (see Fig. \ref{netwoked}). For this purpose, let us state the closed-loop dynamics of $\Sigma_l^{Line}$ as \eqref{Eq:LineCompact} where  
\begin{equation}\label{Eq:linecontroller}
    \Bar{u}_l=\sum_{i\in\mathcal{E}_l}C_{il}x_i,
\end{equation}
with $C_{il}\triangleq \begin{bmatrix}
    \mathcal{B}_{il} & 0 & 0
\end{bmatrix}, \forall l\in\mathcal{E}_i$. 
Note also that $C_{il} = -C_{ti}\bar{C}_{il}^\T$.
By vectorizing \eqref{Eq:linecontroller} over all $l\in\N_L$ and $i\in\N_N$, we get
\begin{equation}\label{Eq:vectorizedlinecontroller}
    \Bar{u}=C x,
\end{equation}
where $C=[C_{il}]_{l\in\mathbb{N}_L,i\in\mathbb{N}_N} \in \R^{L\times  3N}$ (recall that $x \in\R^{3N\times 1}$ and $\bar{u}\in\R^{L\times 1}$). Note also that $C = - \bar{C}^\T C_t$ where $C_t \triangleq \diag(C_{ti}\I_{3}: i\in\N_N) \in \R^{3N\times 3N}$. 


Finally, we need to identify $z$, which represents the vectorized performance output of DC MG (see Fig. \ref{netwoked}). To this end, let us use the output equation in \eqref{Eq:DGCompact}
\begin{equation}\label{Eq:outputperformance}
    z_i = H_i x_i.
\end{equation}
Therefore, vectorizing \eqref{Eq:outputperformance} over all $i\in\N_N$ we get
\begin{equation}\label{Eq:vectorizedoutputperformance}
    z=Hx,
\end{equation}
where $H \triangleq \diag(H_i:i\in\N_N)\in\R^{N\times 3N}$.

Based on \eqref{Eq:DGClosedLoopDynamics_varphi}, \eqref{Eq:vectorizedlinecontroller}, and \eqref{Eq:vectorizedoutputperformance}, and the interconnection relationship
\begin{equation}
    \bm{\tilde{u} \\ \bar{u} \\ z} =  M \bm{x \\ \bar{x} \\ w}, 
\end{equation}
we now can identify the structure of the interconnection matrix $M$ in Fig. \ref{netwoked} as 
\begin{equation}\label{Eq:MMatrix}
    M \triangleq 
    \begin{bmatrix}
        M_{\tilde{u} x} & M_{\tilde{u}\Bar{x}} & M_{\tilde{u} w} \\
        M_{\Bar{u}x} & M_{\Bar{u}\Bar{x}} & M_{\Bar{u}w} \\
        M_{zx} & M_{z\Bar{x}} & M_{zw}\\
    \end{bmatrix} 
    \equiv
    \begin{bmatrix}
        K & \Bar{C} & \I \\
        C & \0 & \0 \\
        H & \0 & \0
    \end{bmatrix}.
\end{equation}

Note that the block matrix $H$ is predefined, and when the physical topology is predefined, so are the block matrices $\Bar{C}$ and $C$ (recall $C = -\bar{C}^\T C_t$). This leaves only the block matrix $K$ inside the block matrix $M$ as a tunable quantity to optimize the desired input-output (from $w$ to $z$) properties of the closed-loop DC MG system. Note that synthesizing $K$ simultaneously determines the distributed global controllers and the communication topology. In the following section, we provide a systemic dissipativity-based approach to synthesize this block matrix $K$.


\section{Dissipativity-Based Control and Topology Co-Design}\label{Passivity-based Control}
This section is organized as follows. First, we introduce the global controller design problem for the networked system (which also determines the optimal communication topology). Next, we present the essential prerequisites for subsystem dissipativity properties to achieve a feasible and efficient global controller design. Following this, we formulate an enhanced local controller design problem to meet the identified prerequisites. Finally, we conclude with a concise overview of the overall design process.

\subsection{Global Controller Synthesis}

Consider a subsystem $\Sigma_i^{DG},i\in\mathbb{N}_N$ (\ref{closedloopdynamic}), which is assumed to be $X_i$-EID with
\begin{equation}\label{Eq:XEID_DG}
    X_i=\begin{bmatrix}
        X_i^{11} & X_i^{12} \\ X_i^{21} & X_i^{22}
    \end{bmatrix}\triangleq
    \begin{bmatrix}
        -\nu_i\mathbf{I} & \frac{1}{2}\mathbf{I} \\ \frac{1}{2}\mathbf{I} & -\rho_i\mathbf{I}
    \end{bmatrix},
\end{equation}
where $\rho_i$ and $\nu_i$ are the passivity indices of $\Sigma_i^{DG}$, i.e., each $\Sigma_i^{DG},i\in\mathbb{N}_N$ is assumed to be IF-OFP($\nu_i,\rho_i$).

Similarly, consider a subsystem $\Sigma_l^{Line},l\in\mathbb{N}_L$ \eqref{Eq:LineCompact}, which is assumed to be $\bar{X}_l$-EID with
\begin{equation}\label{Eq:XEID_Line}
    \bar{X}_l=\begin{bmatrix}
        \bar{X}_l^{11} & \bar{X}_l^{12} \\ \bar{X}_l^{21} & \bar{X}_l^{22}
    \end{bmatrix}\triangleq
    \begin{bmatrix}
        -\bar{\nu}_l\mathbf{I} & \frac{1}{2}\mathbf{I} \\ \frac{1}{2}\mathbf{I} & -\bar{\rho}_l\mathbf{I}
    \end{bmatrix},
\end{equation}
where $\bar{\rho}_l$ and $\bar{\nu}_l$ are the line passivity indices.

Unlike the passivity properties of the DGs, note that we can comment on the passivity properties of lines directly using Prop. \ref{Prop:linear_X-EID} due to the simplicity of the line dynamics \eqref{Eq:LineCompact} (with outputs defined as $y_l = \bar{x}_l$).

\begin{lemma}\label{Lm:LineDissipativityStep}
For each subsystem $\Sigma_l^{Line}, l\in\N_L$ \eqref{Eq:LineCompact}, its passivity indices $\bar{\nu}_l$ and $\bar{\rho}_l$ assumed in (\ref{Eq:XEID_Line}) are such that the LMI problem: 
\begin{equation}\label{Eq:Lm:LineDissipativityStep1}
\begin{aligned}
    \mbox{Find: }\ &\bar{P}_l, \bar{\nu}_l, \bar{\rho}_l\\
    \mbox{Sub. to:}\ &\bar{P}_l > 0, \ \  
    \begin{bmatrix}
        \frac{2\bar{P}_lR_l}{L_l}-\bar{\rho}_l & -\frac{\bar{P}_l}{L_l}+\frac{1}{2}\\
        \star & -\bar{\nu}_l
    \end{bmatrix}\geq0, 
\end{aligned}
\end{equation}
is feasible. The maximum feasible values for $\bar{\nu}_l$ and $\bar{\rho}_l$ respectively are $\bar{\nu}_l^{\max}=0$ and $\bar{\rho}_l^{\max}=R_l$, when $\bar{P}_l =  \frac{L_l}{2}$. 
\end{lemma}
\begin{proof}
According to Prop. \ref{Prop:linear_X-EID}, a line $\Sigma_l^{Line},l\in\N_L$ \eqref{Eq:LineCompact} is $X_i$-EID if and only if there exists $P>0$ such that \eqref{Eq:Prop:linear_X-EID_2} holds with 
$A=\bar{A}_l=-\frac{R_l}{L_l}$, $B=\bar{B}_l=\frac{1}{L_l}$, $C=\bar{C}_l=1$, $D=\bar{D}_l=0$. As the LMI variable $P$ is a scalar, let $P = \bar{P}_l$. From substituting these values, we can state that $\Sigma_l^{Line}$ is $X_i$-EID if and only if there exists $\bar{P}_l, \bar{\nu}_l$ and $\bar{\rho}_l$ such that the LMI problem \eqref{Eq:Lm:LineDissipativityStep1} is feasible. 
Further, from applying the Sylvester criterion to the second constraint in \eqref{Eq:Lm:LineDissipativityStep1}, we get its equivalent conditions: 
\begin{enumerate}
    \item $\frac{2\bar{P}_lR_l}{L_l}-\bar{\rho}_l \geq 0 \iff  \bar{\rho}_l\leq\frac{2\bar{P}_lR_l}{L_l}$;
    \item $-\bar{\nu}_l\geq 0 \iff \bar{\nu}_l \leq 0$; and
    \item $(\frac{2\bar{P}_lR_l}{L_l}-\bar{\rho}_l)(-\bar{\nu}_l)-(-\frac{\bar{P}_l}{L_l}+\frac{1}{2})^2 \geq 0$.
\end{enumerate}
From the first two conditions, we get the maximum values possible for $\bar{\nu}_l=\bar{\nu}_l^{\max} = 0$ and $\bar{\rho}_l = \bar{\rho}_l^{\max} = \frac{2\bar{P}_lR_l}{L_l} = R_l$, respectively. In the latter simplification, the last step is from the above third condition, which implies that $\bar{P}_l = \frac{L_l}{2}$ if $\bar{\nu}_l = 0$. 
\end{proof}

The interconnection matrix $M$ (\ref{Eq:MMatrix}), particularly its block $M_{\tilde{u}x}=K$, can be synthesized by applying the above assumed/proved subsystem EID properties to Prop. \ref{synthesizeM}. By synthesizing $K=[K_{ij}]_{i,j\in\N_N}$, we can uniquely compute the distributed global controller gains $\{(k_{ij}^V,k_{ij}^I,k_{ij}^v):i,j\in\mathbb{N}_N\}$ (\ref{k_ij}) along with the required communication topology.
In designing $K$, we enforce the closed-loop DC MG dynamics be $\textbf{Y}$-EID with $\textbf{Y}$ as stated in Rm. \ref{L2stable}, so as to prevent the amplification of disturbances affecting the voltage regulation performance.
The following theorem formulates the distributed global controller and communication topology co-design problem as an LMI problem, which can be solved using standard convex optimization software toolboxes \cite{Lofberg2004}. 

\begin{theorem}\label{Th:CentralizedTopologyDesign}
The closed-loop dynamics of the DC MG illustrated in Fig. \ref{netwoked} can be made finite-gain $L_2$-stable with an $L_2$-gain $\gamma$ (where $\Tilde{\gamma}\triangleq \gamma^2<\bar{\gamma}$ and $\bar{\gamma}$ is prespecified) from the disturbance input $w(t)$ to the performance output $z(t)$, by synthesizing the interconnection matrix block $M_{\tilde{u}x}=K$ (\ref{Eq:MMatrix}) via solving the LMI problem:
\begin{subequations}
\label{Eq:Th:CentralizedTopologyDesign0}
\begin{align}
&\min_{Q,\{p_i: i\in\N_N\},\{\bar{p}_l: l\in\N_L\},\tilde{\gamma}} &&\sum_{i,j\in\N_N} c_{ij} \Vert Q_{ij} \Vert_1 + c_0 \tilde{\gamma}, \label{Eq:Th:CentralizedTopologyDesign} \\
&\ \ \ \ \ \mbox{ Sub. to:}  
&&p_i > 0, \quad \forall i\in\N_N,\\
& &&\bar{p}_l > 0, \quad \forall l\in\N_L,\\   
& &&0 < \tilde{\gamma} < \bar{\gamma},  
\mbox{ and \eqref{globalcontrollertheorem}},
\end{align}
\end{subequations}
as $K = (\textbf{X}_p^{11})^{-1} Q$. 
Here, the structure of $Q\triangleq[Q_{ij}]_{i,j\in\N_N}$ mirrors that of $K\triangleq[K_{ij}]_{i,j\in\N_N}$ (e.g., the first and third rows are zeros in each block $Q_{ij}$, see \eqref{k_ij}). 
The coefficients $c_0>0$ and $c_{ij}>0,\forall i,j\in\N_N$ are predefined cost coefficients corresponding to the $L_2$-gain and communication links, respectively.
The remaining parameters are defined as follows: $\textbf{X}^{12} \triangleq \diag(-\frac{1}{2\nu_i}\I:i\in\N_N)$, $\textbf{X}^{21} \triangleq (\textbf{X}^{12})^\T$,
$\Bar{\textbf{X}}^{12} \triangleq \diag(-\frac{1}{2\Bar{\nu}_l}\I:l\in\N_L)$
$\Bar{\textbf{X}}^{21} \triangleq (\Bar{\textbf{X}}^{12})^\T$, $\textbf{X}_p^{11} \triangleq \diag(-p_i\nu_i\I:i\in\N_N)$, $\textbf{X}_p^{22} \triangleq \diag(-p_i\rho_i\I:i\in\N_N)$, 
$\Bar{\textbf{X}}_{\bar{p}}^{11} \triangleq \diag(-\bar{p}_l\bar{\nu}_l\I:l\in\N_L)$, 
$\Bar{\textbf{X}}_{\bar{p}}^{22} \triangleq \diag(-\bar{p}_l\bar{\rho}_l\I:l\in\N_L)$, and $\tilde{\Gamma} \triangleq \tilde{\gamma}\I$. 
\end{theorem}
\begin{proof}
The proof follows by considering the closed-loop DC MG (shown in Fig. \ref{netwoked}) as a networked system (shown in Fig. \ref{Networked}) and applying the subsystem dissipativity properties assumed in (\ref{Eq:XEID_DG}) and (\ref{Eq:XEID_Line}) to the interconnection topology synthesis result given in Prop. \ref{synthesizeM}. Note that, we propose the cost function \eqref{Eq:Th:CentralizedTopologyDesign} to jointly optimize the communication topology and robust stability (i.e., $L_2$-gain $\gamma$) of DC MG. The last constraint arises from the specification $\gamma^2 \leq \bar{\gamma}$, which bounds the $L_2$-gain of DC MG. Overall, minimizing the proposed cost function \eqref{Eq:Th:CentralizedTopologyDesign} subject to the LMI constraints listed in \eqref{Eq:Th:CentralizedTopologyDesign0}  optimizes communication topology and robust stability while ensuring the given specification $\gamma^2 \leq \bar{\gamma}$.
\end{proof}

\begin{figure*}[!hb]
\centering
\hrulefill
\begin{equation}\label{globalcontrollertheorem}
\scriptsize
	\bm{
		\textbf{X}_p^{11} & \0 & \0 & Q & \textbf{X}_p^{11}\Bar{C} &  \textbf{X}_p^{11}\\
		\0 & \bar{\textbf{X}}_{\bar{p}}^{11} & \0 & \Bar{\textbf{X}}_{\Bar{p}}^{11}C & \0 & \0\\
		\0 & \0 & \I & H & \0 & \0\\
		Q^\T & C^\T\Bar{\textbf{X}}_{\Bar{p}}^{11} & H^\T & -Q^\T\textbf{X}^{12}-\textbf{X}^{21}Q-\textbf{X}_p^{22} & -\textbf{X}^{21}\textbf{X}_{p}^{11}\bar{C}-C^\T\bar{\textbf{X}}_{\bar{p}}^{11}\bar{\textbf{X}}^{12} & -\textbf{X}^{21}\textbf{X}_p^{11} \\
		\Bar{C}^\T\textbf{X}_p^{11} & \0 & \0 & -\Bar{C}^\T\textbf{X}_p^{11}\textbf{X}^{12}-\bar{\textbf{X}}^{21}\Bar{\textbf{X}}_{\Bar{p}}^{11}C & -\bar{\textbf{X}}_{\bar{p}}^{22} & \0\\ 
		\textbf{X}_p^{11} & \0 & \0 & -\textbf{X}_p^{11}\textbf{X}^{12} & \0 & \tilde{\Gamma}
	}>0 
\end{equation}
\end{figure*}

\subsection{Necessary Conditions on Subsystem Passivity Indices}

Based on the terms $\textbf{X}_p^{11}$, $\textbf{X}_p^{22}$, $\bar{\textbf{X}}_{\bar{p}}^{11}$, $\bar{\textbf{X}}_{\bar{p}}^{22}$, $\textbf{X}^{12}$, $\textbf{X}^{21}$, $\bar{\textbf{X}}^{12}$, and $\bar{\textbf{X}}^{21}$ appearing in \eqref{globalcontrollertheorem} included in the global controller design problem presented in Th. \ref{Th:CentralizedTopologyDesign}, it is clear that the feasibility and the effectiveness of this global controller design depend on the chosen subsystem passivity indices $\{\nu_i,\rho_i:i\in\mathbb{N}_N\}$ \eqref{Eq:XEID_DG} and  $\{\bar{\nu}_l,\bar{\rho}_l:l\in\mathbb{N}_L\}$ \eqref{Eq:XEID_Line} assumed for subsystems \eqref{closedloopdynamic} and \eqref{Eq:LineCompact}, respectively.

However, using Co. \ref{Col.LTI_LocalController_XEID} for designing the local controllers $\{u_{i0}:i\in\mathbb{N}_N\}$ \eqref{Controller}, we can obtain a desired set of subsystem passivity indices for the subsystems \eqref{closedloopdynamic}. Similarly, using Lm. \ref{Lm:LineDissipativityStep}, we can obtain a valid set of passivity indices for the subsystems \eqref{Eq:LineCompact}. While this local controller design and passivity analysis processes can be executed independently from the global controller design \eqref{globalcontrollertheorem}, clearly they can lead to an infeasible and/or ineffective global controller designs.  

Therefore, when designing such local controllers and conducting passivity analysis, one must also consider the specific conditions necessary for the feasibility and effectiveness of the eventual global controller design. The following lemma identifies a few of such conditions based on \eqref{Eq:Th:CentralizedTopologyDesign0} in Th. \ref{Th:CentralizedTopologyDesign}.

\begin{lemma}\label{Lm:CodesignConditions}
For the LMI conditions in \eqref{Eq:Th:CentralizedTopologyDesign0} in Th. \ref{Th:CentralizedTopologyDesign} to hold, it is necessary for the subsystem passivity indices 
$\{\nu_i,\rho_i:i\in\mathbb{N}_N\}$ \eqref{Eq:XEID_DG} and  $\{\bar{\nu}_l,\bar{\rho}_l:l\in\mathbb{N}_L\}$ \eqref{Eq:XEID_Line} assumed for subsystems \eqref{closedloopdynamic} and \eqref{Eq:LineCompact}, respectively, are such that the scalar inequality problem: 
\begin{subequations}\label{Eq:Lm:CodesignConditions}
\begin{align}\nonumber
\mbox{Find: }\ \ &\{(\nu_i,\rho_i,\tilde{\gamma}_i,p_i):i\in\N_N\},&&\{(\Bar{\nu}_l,\bar{\rho}_l,\bar{p}_l):l\in\N_L\}\\
\mbox{Sub. to: } &p_i > 0, &&\forall i\in\N_N, \label{Eq:Lm:CodesignConditions0}\\
&\bar{p}_l>0, &&\forall l\in\N_L, \label{Eq:Lm:CodesignConditions00}\\
&0 < \tilde\gamma_i < \bar{\gamma}, &&\forall i\in\N_N, \label{Eq:Lm:CodesignConditions1}\\
&-\frac{\tilde{\gamma}_i}{p_i}<\nu_i<0,  &&\forall i\in\N_N,  \label{Eq:Lm:CodesignConditions2}\\
&0 < \frac{1}{p_i} < \rho_i,  &&\forall i\in\N_N,    \label{Eq:Lm:CodesignConditions3}\\
&0< \frac{p_i}{4\tilde{\gamma}_i} < \rho_i,  &&\forall i\in\N_N, \label{Eq:Lm:CodesignConditions4} \\
&\Bar{\rho}_l > -\frac{p_i\nu_i}{\Bar{p}_l C_{ti}^2},  &&\forall l\in\mathcal{E}_i, \forall i\in\N_N, \label{Eq:Lm:CodesignConditions5}\\
&\Bar{\rho}_l > \frac{1}{p_i\Bar{p}_l\rho_i}(\frac{p_i}{2C_{ti}}-\frac{\Bar{p}_l}{2})^2,  &&\forall l\in\mathcal{E}_i, \forall i\in\N_N, \label{Eq:Lm:CodesignConditions6} \\
&-\frac{p_i\rho_i}{\Bar{p}_l}<\Bar{\nu}_l<0, &&\forall l\in\mathcal{E}_i , \forall i\in\N_N, \label{Eq:Lm:CodesignConditions7}
\end{align}
\end{subequations}
is feasible.
\end{lemma}

\begin{proof}
Let us define a block diagonal matrix $\hat{\Gamma} \triangleq \diag(\tilde{\gamma}_i\I:i\in\N_N)$ where $\tilde{\gamma}_i \in \R_{\geq 0}, \forall i\in\N_N$. Note that, a necessary condition for the LMI constraint \eqref{globalcontrollertheorem} (note that \eqref{globalcontrollertheorem} is included in \eqref{Eq:Th:CentralizedTopologyDesign0}) can be obtained by selecting $\tilde{\Gamma}$ in  \eqref{globalcontrollertheorem} as $\hat{\Gamma}$, when $\tilde{\Gamma}<\hat{\Gamma}$ (i.e.,  \eqref{globalcontrollertheorem}  $\implies$ \eqref{globalcontrollertheorem} with $\hat{\Gamma}$). 
Note that, $\tilde{\Gamma}<\hat{\Gamma} \iff \tilde{\gamma} \leq \tilde{\gamma}_i, \forall i\in\N_N$. Recall that, we are required to enforce $\tilde{\gamma} \leq \bar{\gamma}$ in \eqref{Eq:Th:CentralizedTopologyDesign0}. However, as
$\tilde{\gamma} \leq \bar{\gamma} \iff \tilde{\gamma} \leq \tilde{\gamma}_i \leq \bar{\gamma}, \forall i\in\N_N$, here we are required to enforce $\tilde{\gamma}_i \leq \bar{\gamma}, \forall i\in\N_N$.

For the feasibility of the LMI constraint in \eqref{globalcontrollertheorem}, the following set of necessary conditions can be identified by examining specific sub-block matrices of the matrix in \eqref{globalcontrollertheorem}:
\begin{subequations}\label{Eq:prooflemma2}
\begin{equation}\label{Eq:Lm:CodesignConditionsStep2_1}
\textbf{X}_p^{11} > 0 , \ \bar{\textbf{X}}_{\bar{p}}^{11} > 0,    
\end{equation}
\begin{equation}\label{Eq:Lm:CodesignConditionsStep2_2}
-Q^\T \textbf{X}^{12}-\textbf{X}^{21}Q-\textbf{X}_p^{22}>0,
\end{equation}
\begin{equation}\label{Eq:Lm:CodesignConditionsStep2_3}
 \Bar{\textbf{X}}_{\Bar{p}}^{22}>0,
 \end{equation}
\begin{equation}\label{Eq:Lm:CodesignConditionsStep2_4}
\bm{\textbf{X}_p^{11} & \textbf{X}_p^{11}  \\
\textbf{X}_p^{11} & \tilde{\Gamma}}>0,
\end{equation}
\begin{equation}\label{Eq:Lm:CodesignConditionsStep2_5}
\bm{\I & H  \\ H^\top & -Q^\T \textbf{X}^{12}-\textbf{X}^{21} Q-\textbf{X}_p^{22}}>0,   \end{equation}
\begin{equation}\label{Eq:Lm:CodesignConditionsStep2_6}
\bm{-Q^\T \textbf{X}^{12}-\textbf{X}^{21} Q-\textbf{X}_p^{22} &  -\textbf{X}^{21}\textbf{X}_p^{11} \\ -\textbf{X}_p^{11} \textbf{X}^{12} & \tilde{\Gamma}}>0,    
\end{equation}
\begin{equation}\label{Eq:Lm:CodesignConditionsStep2_7}
\bm{\textbf{X}_p^{11} &  \textbf{X}_p^{11}\bar{C} \\ \Bar{C}^\T\textbf{X}_p^{11} & - \Bar{\textbf{X}}_{\Bar{p}}^{22}}>0,    
\end{equation}
\begin{equation}\label{Eq:Lm:CodesignConditionsStep2_8}
\bm{-Q^\T \textbf{X}^{12}-\textbf{X}^{21}Q-\textbf{X}_p^{22} & -\textbf{X}^{21}\textbf{X}_p^{11}\bar{C}-C^\T\bar{\textbf{X}}_{\bar{p}}^{11}\bar{\textbf{X}}^{12} \\ -\bar{C}^\T\textbf{X}_p^{11}\textbf{X}^{12}-\bar{\textbf{X}}^{12}\bar{\textbf{X}}_{\bar{p}}^{11}C & -\Bar{\textbf{X}}_{\Bar{p}}^{22} }>0,    
\end{equation}
\begin{equation}\label{Eq:Lm:CodesignConditionsStep2_9}
\bm{\Bar{\textbf{X}}_{\Bar{p}}^{11} & \Bar{\textbf{X}}_{\Bar{p}}^{11}C \\ 
C^\T\Bar{\textbf{X}}_{\Bar{p}}^{11} & -Q^\T\textbf{X}^{12}-\textbf{X}^{21}Q-\textbf{X}_p^{22}}>0.
\end{equation}
\end{subequations}

The necessary conditions for \eqref{Eq:Lm:CodesignConditionsStep2_1} can be expressed as
\begin{align}
&\mbox{\eqref{Eq:Lm:CodesignConditionsStep2_1}}  \iff  -p_i\nu_i > 0 \iff  \nu_i < 0, \quad &&\forall i\in\N_N. \label{Eq:Lm:ProofConditionsStep1}\\ 
&\mbox{\eqref{Eq:Lm:CodesignConditionsStep2_1}}  \iff  -\bar{p}_l\bar{\nu}_l > 0 \iff  \bar{\nu}_l < 0, \quad &&\forall l\in\N_L. \label{Eq:Lm:ProofConditionsStep2}
\end{align}
Similarly, a necessary condition for \eqref{Eq:Lm:CodesignConditionsStep2_2} involves considering its diagonal blocks. Note also that each $Q_{ii}$ block has zeros in its diagonal for any $i\in\N_N$. Using these facts, the necessary conditions for \eqref{Eq:Lm:CodesignConditionsStep2_2} can be obtained as
\begin{align}
\mbox{\eqref{Eq:Lm:CodesignConditionsStep2_2}} &\implies\ 
-Q_{ii}^\T (-\frac{1}{2\nu_i}\I) - (-\frac{1}{2\nu_i}\I)Q_{ii} - (-p_i\rho_i\I) > 0 \nonumber \\
&\implies\ p_i\rho_i > 0 \iff \rho_i > 0, \quad \forall i\in\N_N. \label{Eq:Lm:ProofConditionsStep3}
\end{align}

Using similar arguments, necessary conditions for \eqref{Eq:Lm:CodesignConditionsStep2_3}-\eqref{Eq:Lm:CodesignConditionsStep2_9} can be obtained respectively as
\begin{align}
\mbox{\eqref{Eq:Lm:CodesignConditionsStep2_3}} &\iff  -(-\bar{p}_l\bar{\rho}_{l}\I)>0 \nonumber \\
&\implies \bar{p}_l\bar{\rho}_{l}>0 \iff \bar{\rho}_{l}>0, \quad \forall l\in\N_L. \label{Eq:Lm:ProofConditionsStep4}
\end{align}
\begin{align}\label{Eq:Lm:ProofConditionsStep5}
\mbox{\eqref{Eq:Lm:CodesignConditionsStep2_4}} \iff&\ \bm{-p_i\nu_i & -p_i\nu_i \\ -p_i\nu_i & \tilde{\gamma}_i } > 0 \iff -p_i\nu_i(\tilde{\gamma}_i + p_i\nu_i)  > 0 \nonumber \\ 
\iff&\ \nu_i > -\frac{\tilde{\gamma}_i}{p_i}, \quad \forall i\in\N_N ;
\end{align}
\begin{align}\label{Eq:Lm:ProofConditionsStep6}
\mbox{\eqref{Eq:Lm:CodesignConditionsStep2_5}} \implies&\ \bm{1 & 1 \\ 1 & p_i\rho_i} > 0 \iff \rho_i > \frac{1}{p_i}, \quad \forall i\in\N_N ; 
\end{align}
\begin{align}\label{Eq:Lm:ProofConditionsStep7}
\mbox{\eqref{Eq:Lm:CodesignConditionsStep2_6}} \implies& \bm{p_i\rho_i & -\frac{1}{2}p_i \\ -\frac{1}{2}p_i & \tilde{\gamma}_i} > 0 \iff p_i(\rho_i\tilde{\gamma}_i -\frac{1}{4}p_i) > 0 \nonumber \\
\iff&\ \rho_i > \frac{p_i}{4\tilde{\gamma}_i}, \quad \forall i\in\N_N; 
\end{align}
\begin{align}\nonumber\label{Eq:Lm:ProofConditionsStep8}
\mbox{\eqref{Eq:Lm:CodesignConditionsStep2_7}} 
&\implies
\bm{p_iX_i^{11} &  p_iX_i^{11} \bar{C}_{il} \\ \bar{C}_{il}^\T p_iX_i^{11} & -\bar{p}_l \bar{X}_l^{22}}>0\\
&\iff \bm{-p_i\nu_i & p_i\nu_i\frac{\mathcal{B}_{il}}{C_{ti}} \\ \frac{\mathcal{B}_{il}}{C_{ti}}p_i\nu_i & \Bar{p}_l\Bar{\rho}_l} > 0 \nonumber  \\
&\iff  (-p_i\nu_i\Bar{p}_l\Bar{\rho}_l)-(p_i^2\nu_i^2\frac{\mathcal{B}_{il}^2}{C_{ti}^2})>0 \nonumber \\ 
&\iff 
\Bar{\rho}_l > -\frac{p_i\nu_i}{\Bar{p}_l C_{ti}^2}, \quad \forall l\in\mathcal{E}_i;
\end{align}
\begin{align}\nonumber\label{Eq:Lm:ProofConditionsStep9}
\mbox{\eqref{Eq:Lm:CodesignConditionsStep2_8}} &\implies
\begin{aligned}
\left[ 
\begin{matrix}
    -Q_{ii}^\T (-\frac{1}{2\nu_i}) - (-\frac{1}{2\nu_i})Q_{ii}-p_iX_i^{22}\\-\bar{C}_{il}^\T p_iX_i^{11}X_i^{12}-\bar{X}_l^{12}\bar{p}_l\bar{X}_l^{11}C_{il}
\end{matrix}
\right.\\
\left.
\begin{matrix} 
-X_i^{21}p_iX_i^{11}\bar{C}_{il}-C_{il}^\T\bar{p}_l\bar{X}_l^{11}\bar{X}_l^{12} \\ -\bar{p}_l\bar{X}_l^{22} 
\end{matrix}
\right]>0
\end{aligned}\\
&\iff
\bm{p_i\rho_i & \frac{p_i\mathcal{B}_{il}}{2C_{ti}}-\frac{\mathcal{B}_{il}\Bar{p}_l}{2} \\ \frac{\mathcal{B}_{il}p_i}{2C_{ti}}-\frac{\Bar{p}_l\mathcal{B}_{il}}{2} & \Bar{p}_l\Bar{\rho}_l}>0 \nonumber \\
&\iff p_i\rho_i\Bar{p}_l\Bar{\rho}_l - (\frac{p_i}{2C_{ti}} - \frac{\Bar{p}_l}{2})^2 > 0 \nonumber \\ 
&\iff \Bar{\rho}_l > \frac{1}{p_i\Bar{p}_l\rho_i}(\frac{p_i}{2C_{ti}}-\frac{\Bar{p}_l}{2})^2, \quad \forall l\in\mathcal{E}_i;
\end{align}
\begin{align}\nonumber\label{Eq:Lm:ProofConditionsStep10}
\mbox{\eqref{Eq:Lm:CodesignConditionsStep2_9}} &\implies
\begin{aligned}
\left[ 
\begin{matrix}
    \bar{p}_l\bar{X}_l^{11} \\ C_{il}^\T\bar{p}_l\bar{X}_l^{11}
\end{matrix}
\right.\\
\left.
\begin{matrix} 
\bar{p}_l\bar{X}_l^{11}C_{il} \\ -Q_{ii}^\T (-\frac{1}{2\nu_i}) - (-\frac{1}{2\nu_i})Q_{ii}-p_iX_i^{22} 
\end{matrix}
\right]>0
\end{aligned}\\
&\iff\ \bm{-\Bar{p}_l\Bar{\nu}_l & -\Bar{p}_l\Bar{\nu}_l\mathcal{B}_{il} \\ -\mathcal{B}_{il}\Bar{p}_l\Bar{\nu}_l & p_i\rho_i}>0 \nonumber \\ 
&\iff -\Bar{p}_l\Bar{\nu}_l p_i\rho_i-\Bar{p}_l^2\Bar{\nu}_l^2 \mathcal{B}_{il}^2>0, \nonumber \\
&\iff  \Bar{\nu}_l>-\frac{p_i\rho_i}{\Bar{p}_l}, \quad \forall l\in\mathcal{E}_i.
\end{align}

In conclusion, the necessary conditions in \eqref{Eq:Lm:CodesignConditions0}, \eqref{Eq:Lm:CodesignConditions00}, and \eqref{Eq:Lm:CodesignConditions1} come directly from the conditions in the global controller \eqref{Eq:Th:CentralizedTopologyDesign0}. As shown above, \eqref{Eq:Lm:ProofConditionsStep1} and \eqref{Eq:Lm:ProofConditionsStep5} imply \eqref{Eq:Lm:CodesignConditions2}, \eqref{Eq:Lm:ProofConditionsStep3} and \eqref{Eq:Lm:ProofConditionsStep6} imply \eqref{Eq:Lm:CodesignConditions3}, \eqref{Eq:Lm:ProofConditionsStep3} and \eqref{Eq:Lm:ProofConditionsStep7} imply \eqref{Eq:Lm:CodesignConditions4}, \eqref{Eq:Lm:ProofConditionsStep4} and \eqref{Eq:Lm:ProofConditionsStep8} implies \eqref{Eq:Lm:CodesignConditions5}, \eqref{Eq:Lm:ProofConditionsStep4} and \eqref{Eq:Lm:ProofConditionsStep9} implies \eqref{Eq:Lm:CodesignConditions6}, and \eqref{Eq:Lm:ProofConditionsStep2} and \eqref{Eq:Lm:ProofConditionsStep10} imply \eqref{Eq:Lm:CodesignConditions7}. Basically, \eqref{Eq:Th:CentralizedTopologyDesign0} implies \eqref{Eq:prooflemma2} implies (\eqref{Eq:Lm:ProofConditionsStep1}-\eqref{Eq:Lm:ProofConditionsStep10}) implies \eqref{Eq:Lm:CodesignConditions}.
\end{proof}

In Lm. \ref{Lm:CodesignConditions}, we have considered $p_i, \forall i\in\N_N$ and $\bar{p}_l, \forall l\in\N_L$ as decision variables. This makes \eqref{Eq:Lm:CodesignConditions} a non-linear scalar inequality problem. However, if these variables are assumed to be prespecified, then the scalar inequality problem \eqref{Eq:Lm:CodesignConditions} can be formulated as a linear scalar inequality problem (i.e., an LMI) with the use of appropriate change of variables and convexification steps as detailed in the following corollary.

\begin{corollary}\label{Co:CodesignConditions}
For the LMI conditions in \eqref{Eq:Th:CentralizedTopologyDesign0} in Th. \ref{Th:CentralizedTopologyDesign} to hold, it is necessary that the subsystem passivity indices 
$\{\nu_i,\rho_i:i\in\mathbb{N}_N\}$ and  $\{\bar{\nu}_l,\bar{\rho}_l:l\in\mathbb{N}_L\}$ are such that the LMI problem: 
\begin{subequations}\label{Eq:Co:CodesignConditions}
\begin{align}\nonumber
\mbox{Find: }\ \ &\{(\nu_i,\tilde{\rho}_i,\tilde{\gamma}_i):i\in\N_N\},&&\{(\Bar{\nu}_l,\bar{\rho}_l):l\in\N_L\}\\
\mbox{Sub. to: } 
&0 < \tilde{\gamma}_i < \bar{\gamma}, &&\forall i\in\N_N, \label{Eq:Col:CodesignConditions0}\\
&-\frac{\tilde{\gamma}_i}{p_i}<\nu_i<0,  &&\forall i\in\N_N,  \label{Eq:Col:CodesignConditions1}\\
&0<\tilde{\rho_i}<p_i,  &&\forall i\in\N_N,   \label{Eq:Col:CodesignConditions2}\\
&0<\tilde{\rho_i}<\frac{4\tilde{\gamma}_i}{p_i},  &&\forall i\in\N_N, \label{Eq:Col:CodesignConditions3}\\
&\Bar{\rho}_l > -\frac{p_i\nu_i}{\Bar{p}_l C_{ti}^2},  &&\forall l\in\mathcal{E}_i, \forall i\in\N_N,  \label{Eq:Col:CodesignConditions4}\\
&\Bar{\rho}_l > \frac{\tilde{\rho}_i}{p_i\Bar{p}_l}(\frac{p_i}{2C_{ti}}-\frac{\Bar{p}_l}{2})^2,  &&\forall l\in\mathcal{E}_i, \forall i\in\N_N, \label{Eq:Col:CodesignConditions5}\\
&\bar{\nu}_l > m\tilde{\rho}_i + c, &&\forall l\in\mathcal{E}_i,\forall i\in\N_N, \label{Eq:Col:CodesignConditions6}
\end{align}
\end{subequations} 
is feasible, where  $p_i > 0, \forall i\in\N_N$ and $\bar{p}_l>0, \forall l\in\N_L$ are some prespecified parameters, $\rho_i\triangleq\frac{1}{\tilde{\rho}_i}$, $m\triangleq\frac{\tilde{y}_i^{\max}-\tilde{y}_i^{\min}}{\tilde{\rho}_i^{\max}-\tilde{\rho}_i^{\min}}$, and $c\triangleq\tilde{y}_i^{\max} - m\tilde{\rho}_i^{\max}$ where $\tilde{\rho}_i^{\max}\triangleq \min(p_i,\frac{4\bar{\gamma}_i}{p_i})$, $\tilde{\rho}_i^{\min}=\epsilon$, $\tilde{y}_i^{\max}=-\frac{p_i}{\Bar{p}_l\tilde{\rho}_i^{\max}}$, and $\tilde{y}_i^{\min}=-\frac{p_i}{\Bar{p}_l\tilde{\rho}_i^{\min}}$.
\end{corollary}
\begin{proof}
First, introducing a change of variables: $\rho_i\triangleq\frac{1}{\tilde{\rho}_i}$ for the problem \eqref{Eq:Lm:CodesignConditions}, its inequalities \eqref{Eq:Lm:CodesignConditions3},\eqref{Eq:Lm:CodesignConditions4}, and \eqref{Eq:Lm:CodesignConditions6} can be transformed into inequalities \eqref{Eq:Col:CodesignConditions2}, \eqref{Eq:Col:CodesignConditions3}, and \eqref{Eq:Col:CodesignConditions5}, respectively. On the other hand, the inequalities \eqref{Eq:Col:CodesignConditions0}, \eqref{Eq:Col:CodesignConditions1} and \eqref{Eq:Col:CodesignConditions4} are the same as those in \eqref{Eq:Lm:CodesignConditions1}, \eqref{Eq:Lm:CodesignConditions2} and \eqref{Eq:Lm:CodesignConditions5}, respectively. 
Note also that \eqref{Eq:Lm:CodesignConditions0} and \eqref{Eq:Lm:CodesignConditions00} are now satisfied in \eqref{Eq:Co:CodesignConditions} through the pre-definition of variables $p_i,\forall i\in\N_N$ and $\bar{p}_l,\forall l\in\N_L$.

Finally, note that the remaining inequality \eqref{Eq:Lm:CodesignConditions7}, through the aforementioned change of variables, becomes
\begin{equation}
     -\frac{p_i}{\Bar{p}_l\tilde{\rho_i}}<\Bar{\nu}_l< 0,
\end{equation}
which is non-linear in $\tilde{\rho}_i$. As we will show next, for this non-linear inequality, a necessary linear constraint can be obtained. For this purpose, we use the facts that $\tilde{\rho}_i \in [\tilde{\rho}_i^{\min},\tilde{\rho}_i^{\max}]$ and the concavity of the function $\tilde{y}_i \triangleq -\frac{p_i}{\Bar{p}_l\tilde{\rho}_i}$ with respect to $\tilde{\rho}_i$. 
Consequently, a linear lower bound for $\tilde{y}_i$ can be obtained as:
\begin{equation}
     -\frac{p_i}{\bar{p}_l\tilde{\rho}_i} >  m\tilde{\rho}_i+c,
\end{equation}
with $m$ and $c$ values as given in the lemma statement. Note that, \eqref{Eq:Lm:CodesignConditions7} $\implies$ \eqref{Eq:Col:CodesignConditions6}, as 
\begin{equation}\label{Eq:LinearizedConstraint}
    \bar{\nu}_l > -\frac{p_i}{\bar{p}_l\tilde{\rho}_i} \implies  \bar{\nu}_l > m\tilde{\rho}_i+c. 
\end{equation}

\end{proof}

Note that the scalar non-linear constraints in \eqref{Eq:Lm:CodesignConditions} can be enforced by the formulated scalar linear constraints in \eqref{Eq:Co:CodesignConditions} (in variables $\{(\nu_i,\tilde{\rho}_i,\tilde{\gamma}_i):i\in\N_N\},\{(\Bar{\nu}_l,\bar{\rho}_l):l\in\N_L\}$). 
Note, however, that, \eqref{Eq:Co:CodesignConditions} still is a necessary condition of \eqref{Eq:Lm:CodesignConditions} due to the alternative relaxed linear inequality used for the last non-linear inequality in \eqref{Eq:Lm:CodesignConditions}. To achieve this necessary condition \eqref{Eq:Co:CodesignConditions}, starting with \eqref{Eq:Lm:CodesignConditions} we: (i) employed the change of variables $\tilde{\rho}_i \triangleq \frac{1}{\rho_i}, \forall i\in\N_N$, (ii) relaxed the non-linear bound on $\bar{\nu}_l$ with a linear one based on the limits of $\tilde{\rho}_i$, and (iii) recognized the parameters $p_i, \forall i\in\N_N$ and $\Bar{p}_l, \forall l\in\N_L$ as predefined positive scalar parameters.

In conclusion, using the LMI constraints in \eqref{Eq:Th:CentralizedTopologyDesign0}, we derived the non-linear scalar inequality conditions in \eqref{Eq:Lm:CodesignConditions} as a set of necessary conditions for \eqref{Eq:Th:CentralizedTopologyDesign0}. 
Subsequently, using \eqref{Eq:Lm:CodesignConditions}, we derived the linear scalar inequality constraints in \eqref{Eq:Co:CodesignConditions} as a set of necessary conditions for \eqref{Eq:Lm:CodesignConditions}. Thus, we can conclude that the LMI constraints in  \eqref{Eq:Th:CentralizedTopologyDesign0} imply the linear scalar inequality constraints in \eqref{Eq:Co:CodesignConditions}. In other words, the feasibility of \eqref{Eq:Co:CodesignConditions} is a necessary condition for the feasibility of \eqref{Eq:Th:CentralizedTopologyDesign0}.

\subsection{Local Controller Synthesis}

To enforce the identified necessary LMI constraints in Co. \ref{Co:CodesignConditions} on respective subsystem passivity indices, we next formulate the local controller synthesis problem as an LMI problem. This formulation is summarized in the following theorem.

\begin{theorem}\label{Th:LocalControllerDesign}
Under the predefined: (i) scalar parameters $\{p_i: i\in\N_N\}, \{\bar{p}_l:l\in\N_L\}$, (ii) subsystem parameters $\{C_{ti}: i\in\N_N\}$ in \eqref{DGEQ} and $\{(R_l,L_l):l\in\N_L\}$ in \eqref{line}, and (iii) subsystem matrices $A_i, B_i$ in \eqref{Eq:DG_Matrix_definition}, 
to satisfy the necessary conditions of \eqref{Eq:Th:CentralizedTopologyDesign0} identified in Co. \ref{Co:CodesignConditions}, the local controller gains $\{K_{i0}, i\in\N_N\}$ (\ref{Controller}) and passivity indices should be determined by solving the LMI problem:
\begin{subequations}
\begin{align}\nonumber
&\mbox{Find: }\ \{(\tilde{K}_{i0}, P_i, \nu_i, \tilde{\rho}_i, \tilde{\gamma}_i):i\in\mathbb{N}_N\}, \{(\bar{P}_l, \bar{\nu}_l,\bar{\rho}_l):l\in\mathbb{N}_L\} \\
&\mbox{Sub. to: }\ P_i > 0,  \forall i\in\mathbb{N}_N, \ \bar{P}_l > 0, \forall l\in\mathbb{N}_L \\
&\ 
\bm{\tilde{\rho}_i\I & P_i & \0 \\
P_i &-\mathcal{H}(A_iP_i + B_i\Tilde{K}_{i0})& -\I + \frac{1}{2}P_i\\
\0 & -\I + \frac{1}{2}P_i & -\nu_i\I} > 0, \forall i\in\mathbb{N}_N \label{GetPassivityIndicesDGs}\\
&\ 
\bm{
    \frac{2\bar{P}_lR_l}{L_l}-\bar{\rho}_l & -\frac{\bar{P}_l}{L_l}+\frac{1}{2}\\
        \star & -\bar{\nu}_l
    }\geq0, \forall l\in\mathbb{N}_L \label{GetPassivityIndicesLines}\\
&\ -\frac{\tilde{\gamma}_i}{p_i}<\nu_i<0, \forall i\in\N_N \label{Eq:Th:LocalControllerDesign1}\\
&\ 0<\tilde{\rho_i}< \min\left\{p_i,\frac{4\tilde{\gamma}_i}{p_i}\right\}, \forall i\in\N_N \label{Eq:Th:LocalControllerDesign2}\\
&\Bar{\rho}_l > \max\left\{-\frac{p_i\nu_i}{\Bar{p}_l C_{ti}^2},\frac{\tilde{\rho}_i}{p_i\Bar{p}_l}(\frac{p_i}{2C_{ti}}-\frac{\Bar{p}_l}{2})^2 \right\}>0, \forall l\in\mathcal{E}_i, \forall i\in\N_N \label{Eq:Th:LocalControllerDesign3}\\
&\bar{\nu}_l > m\tilde{\rho}_i + c, \forall l\in\mathcal{E}_i,\forall i\in\N_N \label{Eq:Th:LocalControllerDesign4}
\end{align}
\end{subequations}
where $K_{i0} \triangleq \tilde{K}_{i0}P_i^{-1}$, $\rho_i\triangleq\frac{1}{\tilde{\rho}_i}$, and the parameters $m$ and $c$ are the same as those defined in Co. \ref{Co:CodesignConditions}.
\end{theorem} 
\begin{proof}
The proof proceeds as follows: (i) We start by considering the dynamic models of $\Sigma_i^{DG}$ and $\Sigma_l^{Line}$ as described in (\ref{closedloopdynamic}) and \eqref{Eq:LineCompact}, respectively. We then apply an LMI-based controller synthesis and analysis technique to enforce and identify the subsystem passivity indices assumed in (\ref{Eq:XEID_DG} and \ref{Eq:XEID_Line}), respectively. (ii) Next, we apply the LMI problems in Co. \ref{Col.LTI_LocalController_XEID} and Lm. \ref{Lm:LineDissipativityStep} to get the passivity indices of $\Sigma_i^{DG}$ and $\Sigma_l^{Line}$, which leads to the constraints \eqref{GetPassivityIndicesDGs} and \eqref{GetPassivityIndicesLines}, respectively. (iii) Finally, we impose the necessary conditions on subsystem passivity indices identified in Co. \ref{Co:CodesignConditions} to support the feasibility and effectiveness of the global control and communication topology co-design approach presented in (\ref{globalcontrollertheorem}). In particular, the constraints \eqref{Eq:Th:LocalControllerDesign1}, \eqref{Eq:Th:LocalControllerDesign2}, \eqref{Eq:Th:LocalControllerDesign3} and \eqref{Eq:Th:LocalControllerDesign4} are derived respectively from constraints \eqref{Eq:Col:CodesignConditions1}, \eqref{Eq:Col:CodesignConditions2} and \eqref{Eq:Col:CodesignConditions3}, \eqref{Eq:Col:CodesignConditions4} and \eqref{Eq:Col:CodesignConditions5}, and \eqref{Eq:Col:CodesignConditions6} in Co. \ref{Co:CodesignConditions}.
\end{proof}


\subsection{Overview}
To conclude this section, we outline the main steps
necessary for designing local controllers, global controllers,
and communication topology in a centralized manner for the
considered DC MG.

First, we need to specify the parameters of the DC MG (i.e., electrical characteristics) and then determine the state and input matrices of the DGs and lines, represented as $A_i, B_i, \forall i\in\mathbb{N}_N$ \eqref{Eq:DG_Matrix_definition} and $\bar{A}_l, \bar{B}_l, \forall l\in\mathbb{N}_L$ \eqref{Eq:Line_Matrix_definition}. Next, we select some scalar parameters for $p_i, \forall i \in \mathbb{N}_N$ and $\bar{p}_l, \forall l \in \mathbb{N}_L$. Following this, we synthesize the local controllers in Th. \ref{Th:LocalControllerDesign} and obtain the passivity indices of the DGs $\{\rho_i, \nu_i : \forall i \in \mathbb{N}_N\}$ \eqref{Eq:XEID_DG} and lines $\{\bar{\rho}_l, \bar{\nu}_l : \forall l \in \mathbb{N}_L\}$ \eqref{Eq:XEID_Line}. Finally, we use these passivity indices to synthesize the global controller and communication topology of the DG MG using Th. \ref{Th:CentralizedTopologyDesign}. This process is illustrated in Fig. \ref{fig.overview}, providing a clear and comprehensive roadmap for implementation. 

\begin{figure}
    \centering
    \includegraphics[width=0.7\columnwidth]{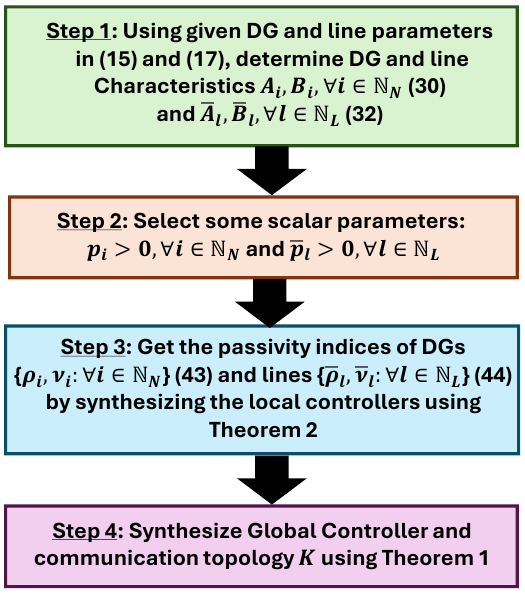}
    \caption{Overview of the procedure for designing local and distributed global controllers for the considered DC MG.}
    \label{fig.overview}
\end{figure}

\section{Simulation Results}\label{Simulation}
To demonstrate the effectiveness of the proposed dissipativity-based control and topology co-design technique, we simulated two islanded DC MG configurations under different scenarios in a MATLAB/Simulink environment. In Case I, we used a DC MG consisting of 4 DGs and 4 lines, while in Case II, we used a DC MG consisting of 6 DGs and 9 lines. In each case, the physical topology of the DC MG was randomly generated using random geometric graph generation technique with connectivity parameter as 0.6 \cite{zhang2018robustness}. Each DG also supports a corresponding local load. The reference voltage amplitude was set to 48 V, and the nominal voltage of voltage sources was set to 120 V. The parameters (in particular, their mean nominal values) of the used DC MG components are listed in Table \ref{table.parameters}. Simply, each DG, load and transmission line parameter was selected by adding a random variation to the correponding mean nominal parameter value given in Table \ref{table.parameters} so as to create challenging heterogeneous DC MG configurations. Fig. \ref{fig.physicaltopology} shows an example physical topology of an islanded DC MG. 

Each considered DC MG configuration,
designed using the proposed dissipativity-based control and topology co-design technique (or a conventional droop control technique), 
was tested under a sequence of load changes. In particular, at $t=3$ s, a constant current load $\bar{I}_L$ was added to the DGs, and an existed constant impedance load $Y_L$ was decreased at $t=4$ s and increased again at $t=7$ s. It is important to note that the constant current load $\bar{I}_L$ and the exogenous reference voltage input $V_r$ were incorporated into the simulation to simulate disturbances, as described in \eqref{statespacewithcontroller}. Also, note that per-unit currents represent normalized currents relative to their capacity, calculated by dividing the output currents by their nominal currents.

Furthermore, for each DC MG case, when the proposed control and topology co-design process is used (given in Th. \ref{Th:CentralizedTopologyDesign}), we considered two different ways to influence the resulting communication topology.
In the first scenario, we include an additional constraint in the co-design process so as to restrict the resulting communication topology to exactly match the physical topology (but with bi-directional links). Therefore, this scenario basically designs the communication topology under a ``hard graph constraint.'' In the second scenario, we omit such hard graph constraints but penalize the use of communication links proportionally to their physical lengths via appropriately selecting the cost coefficients in the communication cost objective used in the co-design process. Therefore, this scenario basically designs the communication topology under ``soft graph constraints.'' In the proposed co-design process (see Step 2 in Fig. \ref{fig.overview}), we selected the design parameters as $p_i=0.01,\forall i$ and $p_l=0.001, \forall l$.   

\begin{figure}
    \centering
    \includegraphics[width=0.9\columnwidth]{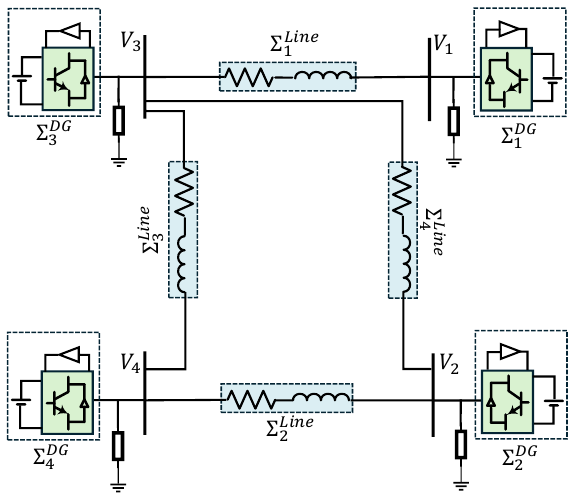}
    \caption{The physical topology of DC MG with 4 DGs and 4 Lines.}
    \label{fig.physicaltopology}
\end{figure}

\begin{table}
\caption{Parameters of the simulated islanded DC MGs\label{table.parameters}}
\centering
\renewcommand{\arraystretch}{1}

\makeatletter
\def\thickhline{%
  \noalign{\ifnum0=`}\fi\hrule \@height \thickarrayrulewidth \futurelet
   \reserved@a\@xthickhline}
\def\@xthickhline{\ifx\reserved@a\thickhline
               \vskip\doublerulesep
               \vskip-\thickarrayrulewidth
             \fi
      \ifnum0=`{\fi}}
\makeatother

\newlength{\thickarrayrulewidth}
\setlength{\thickarrayrulewidth}{2.5\arrayrulewidth}

\begin{tabular}{c | c}
  \thickhline
\textbf{DC MG Parameters} & \textbf{Mean Nominal Values} \\
   \hline \hline
    Internal Resistance &  $R_t$ = 50 $m\Omega$ \\
\hline
 Internal Inductance &   $L_t$ = 10 mH \\
 \hline
 Filter Capacitance &   $C_t$ = 0.5 F \\
 \hline
 Constant Impedance Load &   $Y_L$ = 2 $\Omega$\\
 \hline
 Constant Current Load &   $\bar{I}_{L}$ = 0.5 A\\
 \hline
 Line Resistance &   $R_l$ = 50 $\Omega$\\
 \hline
 Line Inductance &   $L_l$ = 10 mH \\
 \hline
 Reference Voltage & $V_{r}$ = 48 V \\
 \hline
 Voltage Source & 120 V\\
 \thickhline
\textbf{Co-design Parameters} & \textbf{Selected Values}
 \\ \hline \hline
 
$p$ & 0.01 \\
\hline
$\bar{p}$ & 0.001\\
\hline
$L_2$-gain bound $\bar{\gamma}$ & 1000\\
\thickhline
\textbf{Passivity Indices} & \textbf{Results}
 \\ \hline \hline
DG Output Passivity $\rho$ & 5049\\
\hline
DG Input Passivity $\nu$ & -26000 \\
\hline
Line Output Passivity $\bar{\rho}$ & $2.3\times 10^9$ \\
\hline
Line Input Passivity $\bar{\nu}$ & $-9.8\times 10^5$\\
\thickhline
\end{tabular}
\end{table}

\subsection{Case I: DC MG with 4 DGs and 4 Lines}\label{sec.proposedcontroller}
In this section, we demonstrate the effectiveness of the proposed dissipativity-based control and topology co-design technique when applied to a DC MG with 4 DGs and 4 lines. Figures \ref{fig.communicationtopology}(a) and  \ref{fig.communicationtopology}(b) show the obtained communication topologies respectively under hard and soft graph constraints. Note that, under soft (as opposed to hard) graph constraints, the resulting total number of communication links from the proposed co-design process has decreased by $25\%$. Figure \ref{fig.voltagecurrent} shows the observations from the DC MG co-designed under soft graph constraints. It is worth noting that, identical observations to those in Fig. \ref{fig.voltagecurrent} were observed from the DC MG co-designed under hard graph constraints. These observations imply that DC MG communication topologies can be optimized without compromising performance.

\begin{figure}
    \centering
    \includegraphics[width=0.9\columnwidth]{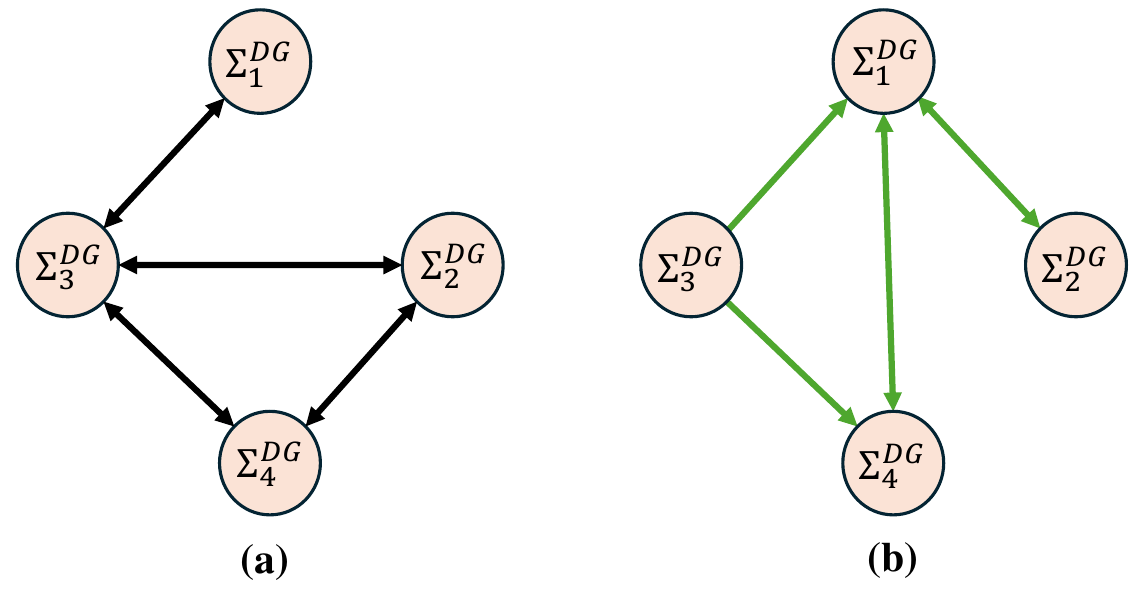}
    \caption{The communication topology for DC MG shown in Fig. \ref{fig.physicaltopology} under (a) hard and (b) soft graph constraints.}
    \label{fig.communicationtopology}
\end{figure}

Figure \ref{fig.voltagecurrent} shows the voltages and per-unit currents observed for each DG. As seen in Fig. \ref{fig.voltagecurrent}(a), the voltages smoothly track the reference value of 48 V. Figure \ref{fig.voltagecurrent}(b) displays the per-unit currents, where each DG tracks a value of 1, indicating that all DGs are operating at their rated capacity, ensuring effective load sharing. During the load changes, it can be seen that the proposed controller successfully restores voltages to the reference value $V_r$ and maintains balanced current sharing among the DGs.

\begin{figure}
    \centering
    \includegraphics[width=1\columnwidth]{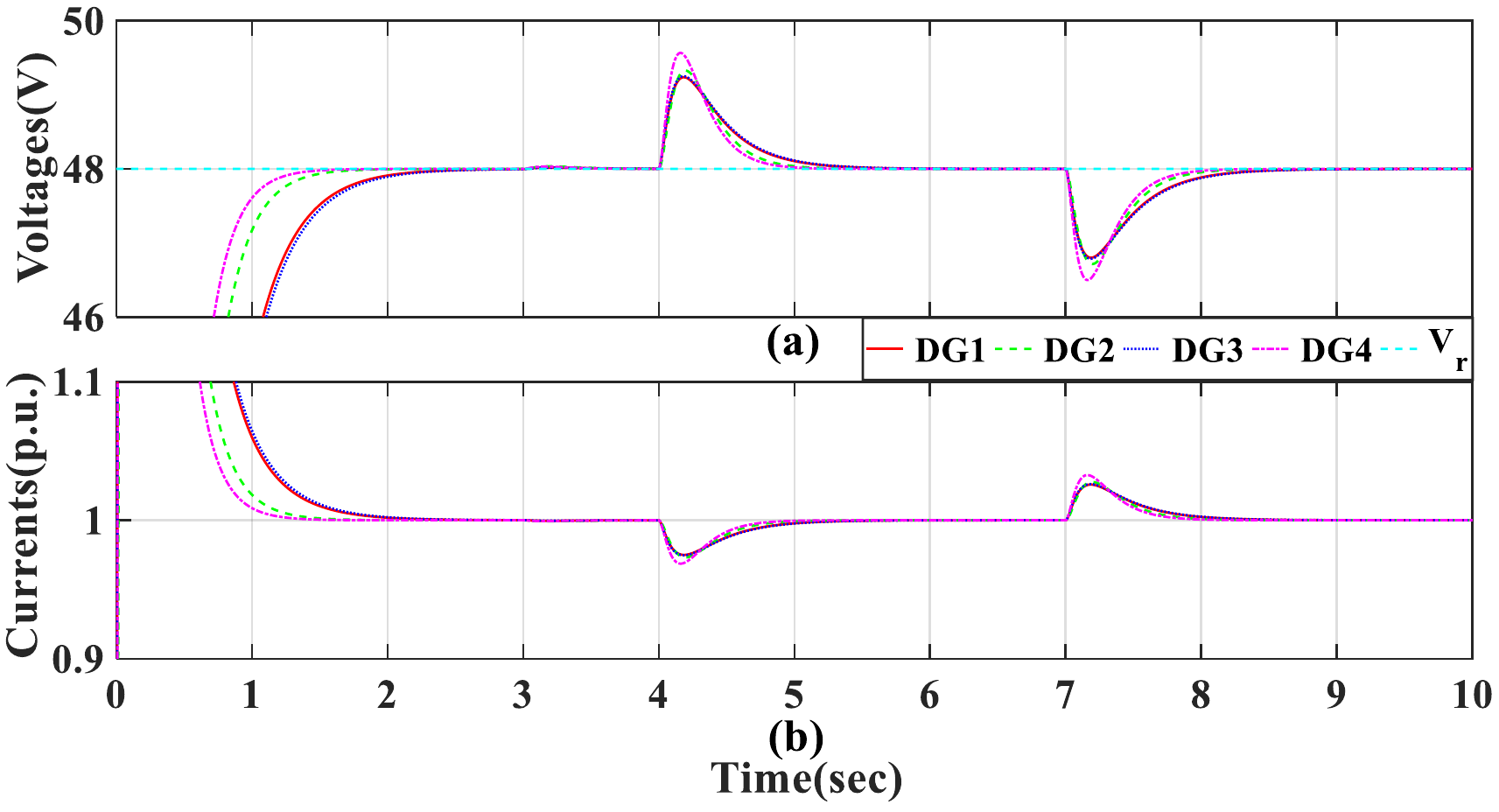}
    \caption{The DG output (a) voltages and (b) per-unit currents using the proposed dissipativity-based distributed controllers in DC MG shown in Fig. \ref{fig.physicaltopology}.}
    \label{fig.voltagecurrent}
\end{figure}




As demonstrated, the proposed co-design approach eliminates the need for manual design of the communication topology, allowing for optimized use of communication links without compromising the voltage regulation performance and robustness. In this case study, we observed that the soft graph constraint scenario resulted in requiring fewer communication links than in the hard graph constraints scenario, thus reducing communication bandwidth. It is worth mentioning that this reduction is achieved while ensuring key performance metrics, such as voltage regulation and load sharing, remain unaffected by topology change.

\subsection{Case II: DC MG with 6 DGs and 9 Lines}
To further demonstrate the effectiveness of the proposed dissipativity-based controller, we extended the simulation to a larger and more complex DC MG configuration that includes 6 DGs and 9 lines. The electrical and communication topologies of the DC MG are shown in Fig. \ref{fig.scalability} and Fig. \ref{fig.scalablity_communication}, respectively.

\begin{figure}
    \centering
    \includegraphics[width=0.9\columnwidth]{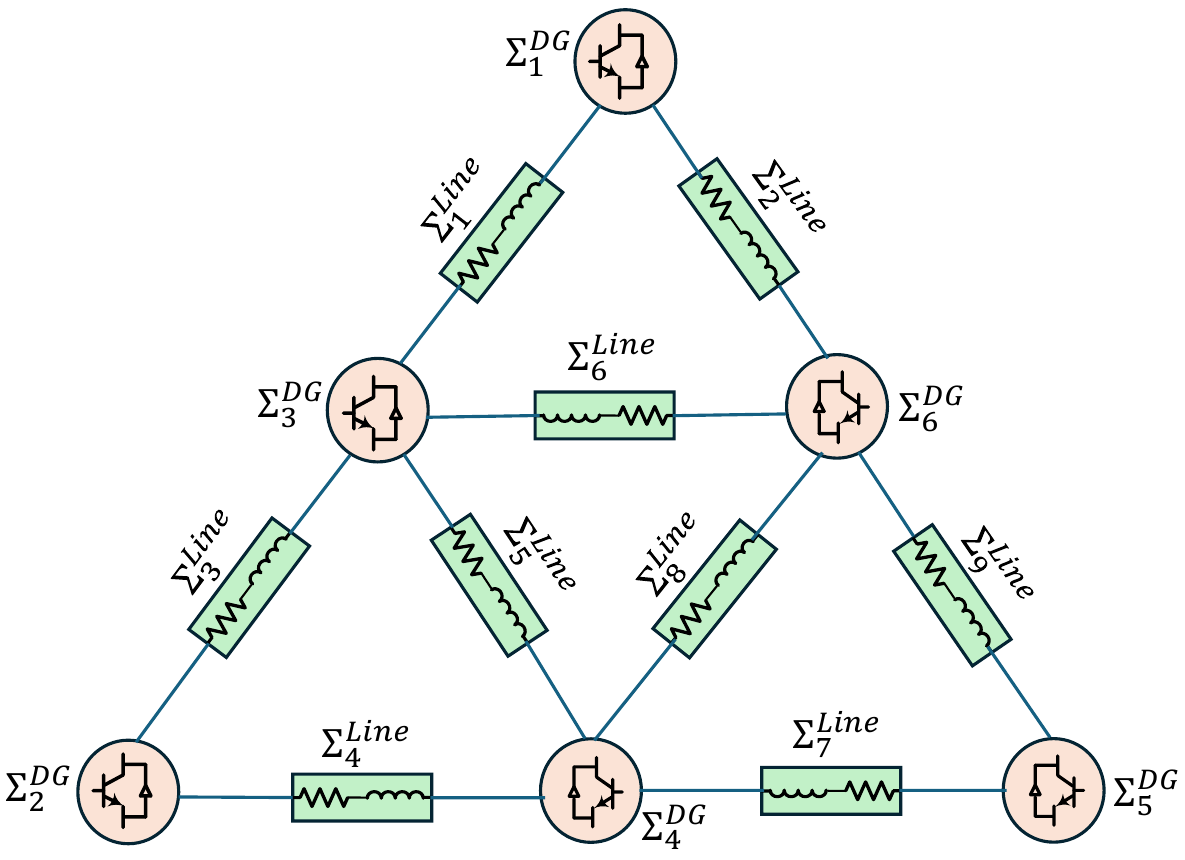}
    \caption{The physical topology of DC MG with 6 DGs and 9 transmission lines.}
    \label{fig.scalability}
\end{figure}

As shown in Fig. \ref{fig.6DGs}, using the communication topology obtained under soft graph constraints, the controllers successfully restores the output voltages to the reference value of 48 V while ensuring proper load sharing among the 6 DGs. Like Case I, we compared the co-design solutions obtained under hard and soft graph constraints. As shown in Fig. \ref{fig.scalablity_communication}, under the soft graph constraints, the resulting communication topology involves $39\%$ decreased number of communication links, while still achieving identical performance in voltage regulation and load sharing.

\begin{figure}
    \centering
    \includegraphics[width=1\columnwidth]{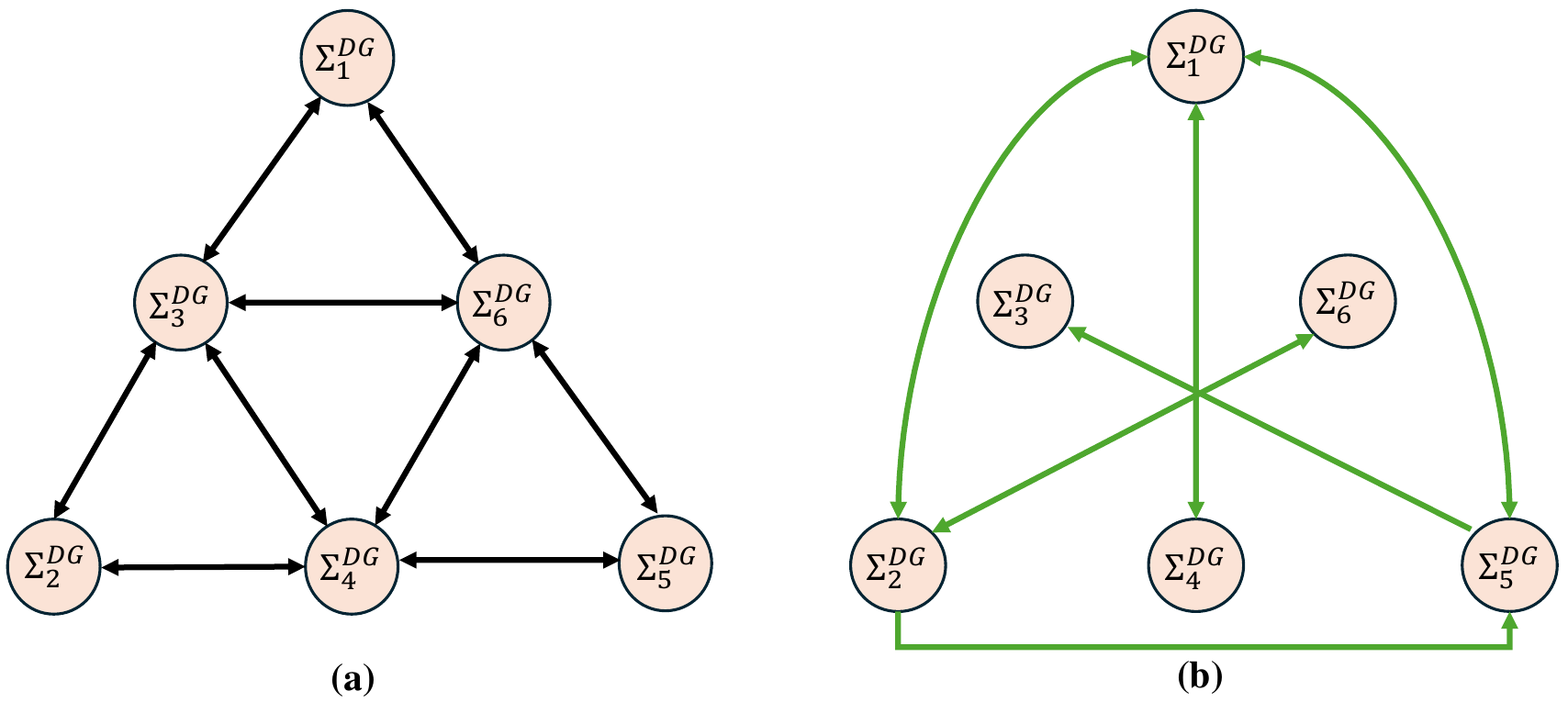}
    \caption{The communication topology for DC MG shown in Fig. \ref{fig.scalability} under (a) hard and (b) soft graph constraints.}
    \label{fig.scalablity_communication}
\end{figure}

\begin{figure}
    \centering
    \includegraphics[width=1\columnwidth]{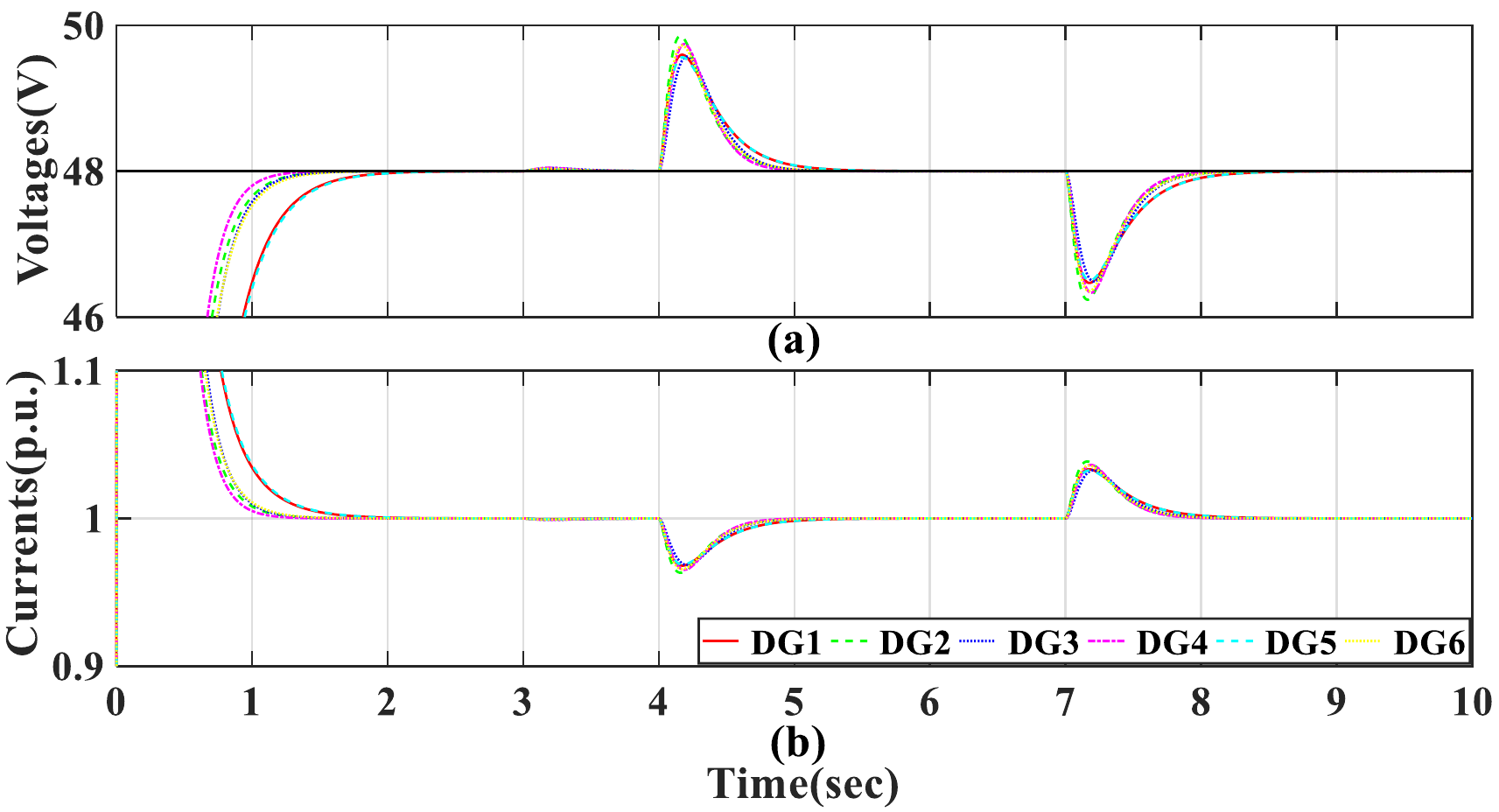}
    \caption{The DG output (a) voltages and (b) per-unit currents using the proposed dissipativity-based distributed controllers in DC MG shown in Fig. \ref{fig.scalability}.}
    \label{fig.6DGs}
\end{figure}

\subsection{Comparison with Traditional Droop Controller}
Finally, we compared the performance of the proposed dissipativity-based controller with a conventional droop controller in terms of voltage regulation performance in the DC MG configuration considerd in Case I. To ensure a fair comparison, our dissipativity-based co-design process used hard graph constraints (leading to a communication topology identical to the physical topology). Additionally, all DGs and line parameters were kept consistent across both methods, and the droop controller gains were tuned to optimize the resulting performance, starting with the droop control gains sourced from the well-established work  \cite{guo2018distributed}.

As demonstrated in Fig. \ref{fig.droop_control}, the proposed dissipativity-based controller maintained precise voltage tracking at 48 V throughout the simulation, with minimal overshoot during load changes. In contrast, the droop controller experienced more significant voltage overshoots at times of load changes, highlighting its less effective response to load variations. Furthermore, the droop control method initially exhibits a voltage drop, which is only compensated by activating a secondary control \cite{Najafirad2} at $t=1$ s to restore the voltages to the reference value gradually.

This comparison underscores the limitations of the droop control approach, which requires additional control layers, i.e.,  secondary or distributed control, to compensate for both voltage deviations and overshoot. The proposed dissipativity-based controller, however, achieves stable and robust voltage regulation without the need for such extra control layers, resulting in smoother voltage restoration and improved performance during load changes. These results demonstrate the superiority of the proposed controller over conventional droop controllers, making it a more effective solution for DC MG applications.

\begin{figure}
    \centering
    \includegraphics[width=0.9\columnwidth]{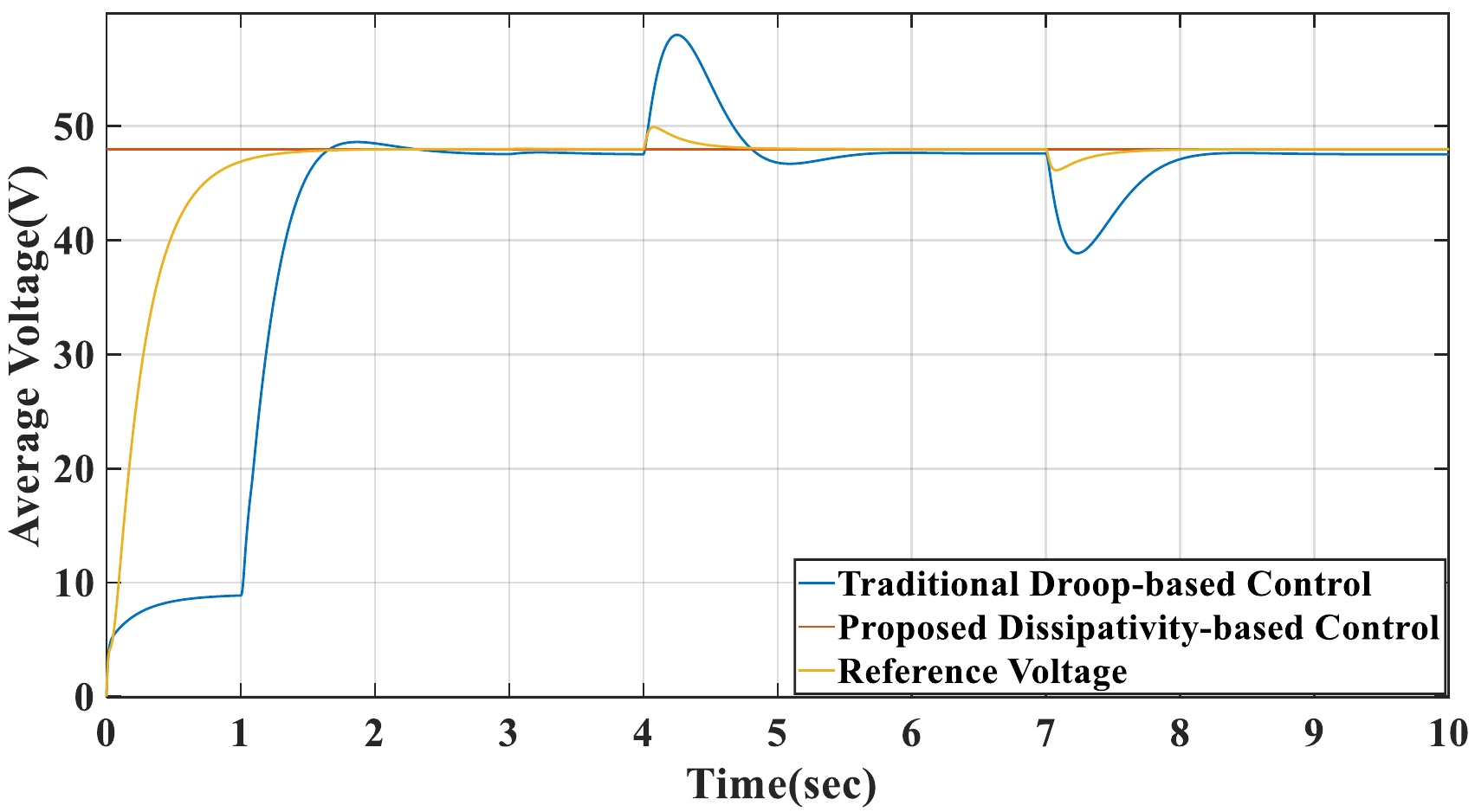}
    \caption{Comparison of average voltage regulation between the proposed dissipativity-based controller and droop controllers for DC MG.}
    \label{fig.droop_control}
\end{figure}

\section{Conclusion}\label{Conclusion}
This paper proposes a dissipativity-based distributed droop-free control and communication topology co-design approach for DC MGs, where DGs and transmission lines are treated as two interconnected sets of subsystems through an interconnection matrix. By leveraging dissipativity theory, we develop a unified framework that simultaneously addresses the distributed global controller design and the communication topology design problems. To support the feasibility of this global co-design process, we use specifically designed local controllers at the subsystems. 
To enable efficient and scalable evaluations, we formulate all design problems as LMI-based convex optimization problems. This integrated approach ensures robust voltage regulation and current sharing with respect to various disturbances using a distributed droop-free controller over an optimized communication topology. Future work will focus on developing the co-design approach in a decentralized and compositional manner, enabling plug-and-play and generalizing this co-design approach to AC MGs.

\bibliographystyle{IEEEtran}
\bibliography{References}

\end{document}